\documentclass[12pt,draftcls,peerreview,onecolumn]{IEEEtran}

\usepackage{amsmath, graphics,amssymb,epsfig,subfigure,color,amsthm}

\usepackage{algorithm}
\usepackage{algorithmic}
\usepackage{float}
\usepackage{multirow}
\usepackage{cuted,flushend}
\usepackage{midfloat}
\newtheorem{theorem}{Theorem}

\newtheorem{lemma}{Lemma}

\newtheorem{corollary}{Corollary}

\usepackage{cite}

\def\b{\mathbb}

\def\b{\mathbb}

\def\d{{\rm d}}

\def\sinc{{\rm sinc}}

\begin{document}
\title{ Spectral Efficiency Scaling Laws in\\ Dense Random Wireless Networks with\\ Multiple Receive Antennas }
\author{\normalsize  Namyoon Lee, Francois Baccelli, and Robert W. Heath Jr. \bigskip
\\
\normalsize Wireless Networking and Communications Group\\
\normalsize Department of Electrical and Computer Engineering
\\ \normalsize The University of Texas at
Austin, Austin, TX 78712 USA\\
      { \normalsize E-mail~:~\{namyoon.lee, rheath\}@utexas.edu}, baccelli@math.utexas.edu  }

\date{}
\maketitle
\begin{abstract}
This paper considers large random wireless networks where transmit-and-receive node pairs  communicate within a certain range while sharing a common spectrum. By modeling the spatial locations of nodes based on stochastic geometry, analytical expressions for the ergodic spectral efficiency of a typical node pair are derived as a function of the channel state information available at a receiver (CSIR) in terms of relevant system parameters: the density of communication links, the number of receive antennas, the path loss exponent, and the operating signal-to-noise ratio. One key finding is that when the receiver only exploits CSIR for the direct link, the sum of spectral efficiencies linearly improves as the density increases, when the number of receive antennas increases as a certain super-linear function of the density. When each receiver exploits CSIR for a set of dominant interfering links in addition to the direct link, the sum of spectral efficiencies linearly increases with both the density and the path loss exponent if the number of antennas is a linear function of the density. This observation demonstrates that having CSIR for dominant interfering links provides a multiplicative gain in the scaling law. It is also shown that this linear scaling holds for direct CSIR when incorporating the effect of the receive antenna correlation, provided that the rank of the spatial correlation matrix scales super-linearly with the density. Simulation results back scaling laws derived from stochastic geometry.
\end{abstract}
 
%%%%%%%%%%%%%%%%%%%%%%%%%%%%%%%%%%%%%%%%%%%%%%%%%%%%%%%%%%%%%%%%%%%%%%%%%%%%%%%%%%%%%%%%%%%%%%%%%%
\section{Introduction}

There is an increasing need for direct communication  between wireless device pairs to support proximity-based social networking
applications or media sharing \cite{Doppler1,Corson,Lin
}. As the number of device pairs increases, the coexistence of multiple communication links in the shared spectrum is however challenging due to mutual interference, which poses fundamental limitations on the throughput. One of the main difficulties is that in many cases, such as device-to-device (D2D) \cite{Doppler1,Corson,Lin,Doppler_PC,Namyoon_PC,Geometry_scheduling,FlahsLinQ,ITLinQ} and mobile ad hoc network \cite{Weber1,Baccelli,Hasan,Nguyen,Weber_SIC,Jindal_SIC_JD,Zhang_SIC}, the communication links cannot be coordinated in a centralized way due to the amount of signaling overhead associated with coordination. This has raised the need for distributed interference management with low signaling overheads.

Two main distributed interference management approaches have been proposed in the context of such networks: 1) distributed power control techniques and 2) distributed link scheduling algorithms. In \cite{Doppler_PC}, simple yet heuristic power control methods were proposed to regulate transmit power to mitigate interference between links. Optimal distributed on-off power control strategies were proposed to maximize the transmission capacity \cite{Jindal_PC}, coverage probability \cite{Zhang_PC}, and spectral efficiency for D2D networks \cite{Namyoon_PC}. The main limitation in\cite{ Doppler_PC,Jindal_PC,Zhang_PC,Namyoon_PC
} is that the power control methods are only effective when the number of links per unit area is small.

Distributed link scheduling has also recently received much attention. In the context of ad hoc and wireless local area networks, various distributed scheduling mechanisms for interference management  have been proposed in the literature, such as ALOHA type medium access control (MAC) protocols (e.g., \cite{Weber1,Baccelli,Hasan}), random sequential adsorption MAC protocols \cite{Nguyen}, and distributed scheduling by channel thresholding \cite{Weber2}. The main limitation of these approaches is the inefficient network spatial packing  resulting from the underlying interference avoidance strategies. By leveraging interference cancellation techniques at the receiver, advanced distributed scheduling mechanisms have also been proposed to increase the spatial packing performance in\cite{ Weber_SIC,Jindal_SIC_JD,Zhang_SIC}.  

Recently, more sophisticated distributed scheduling mechanisms were proposed in the context of D2D networks \cite{Geometry_scheduling,FlahsLinQ,ITLinQ}. In \cite{Geometry_scheduling}, a geometric scheduling method was  proposed where the exclusion regions  between different D2D links are created based on link geometries. A signal-to-interference ratio (SIR) based
distributed scheduling method called FlashLinQ was  proposed in \cite{FlahsLinQ}, where the exclusion regions are dynamically created based on link priorities and SIRs. This scheduling algorithm was shown to provide a better throughput than that of preexisting MAC protocols. Leveraging the optimality condition of treating interference as noise (TIN) in \cite{TIN_paper}, an information theoretic independent set scheduling algorithm was proposed called ITLinQ \cite{ITLinQ}, which achieves optimal sum rate performance for constant rate loss.
More elaborate distributed scheduling mechanisms in \cite{Geometry_scheduling,FlahsLinQ,ITLinQ}
may appear to yield much higher throughput, but the induced communication overheads in handshaking processes need to be subtracted, and the net gain compared to a simple ALOHA scheduling method may not be large enough when the density of node pairs is  sparse.

In this paper, we use multiple antennas to perform distributed interference management \cite{Jindal_PZF, Huang,Hunter_MIMO,Akoum}. Multi-antenna communication techniques provide an effective approach to mitigate interference because of their large gains in terms of channel capacity and reliability. In the context of ad hoc networks modeled by stochastic geometry, upper and lower bounds were obtained on the transmission capacity when multiple antennas are employed at transceivers in \cite{Jindal_PZF, Huang,Hunter_MIMO,Akoum}. In particular, interference cancellation techniques using multiple receive antennas were shown to substantially increase the transmission capacity of ad hoc networks  \cite{Jindal_PZF,Huang}. For example, by leveraging the idea of partial zero-forcing in \cite{Jindal_PZF}, it was shown that the transmission capacity increases with the node density linearly using the multiple receive antennas. Continuing in the same spirit yet with a different perspective, we analyze the benefits of using multiple antennas at receivers from a spectral efficiency point-of-view. Unlike the transmission capacity that measures the spatial density of successful transmissions per unit area, subject to a given outage probability constraint, in this paper, we consider the \emph{ergodic spectral efficiency} as a performance metric. The key limitation of transmission capacity is that the rate target is fixed, implying that the rate adaptation techniques cannot be applied over different fading realizations. Whereas, the latter measures the achievable Shannon transmission rates per unit area that averaging the rate over the different fading realizations. Arguably, this quantity is more appropriate than the transmission capacity in contemporary wireless systems where a coded packet is transmitted over multiple fading realizations \cite{Angel_Jindal}.  

We consider a dense wireless network whose topology is modeled by means of a homogeneous Poisson point process (PPP) with node density $\lambda$. Such a random PPP model captures the irregular spatial structure of mobile node locations and helps to analytically quantify the interference. We summarize our main contributions as follows:

\begin{itemize}
\item As a starting point, we first consider the case where each receiver exploits CSIR for the direct link. Applying  maximum ratio combining (MRC)\cite{Tse}, we derive an exact analytical expression for the ergodic spectral efficiency in the network as a function of 1) the density of wireless links $\lambda$, 2) the number of receive antennas $N_{\rm r}$, 3) the path loss exponent $\alpha$, and 4) the operating signal-to-noise ratio (SNR). By deriving a tight lower and upper bound on the sum spectral efficiency, we show that the ergodic spectral efficiency scales with respect to the density as $\Theta(\lambda\log_{2}\left(1+\lambda^{\beta-\frac{\alpha}{2}}\right))$ when $N_{\rm r} =c\lambda^{\beta}$ with some $c>0$ and $\alpha>2$.

\item Next, we consider the case in which each receiver has perfect knowledge of the CSIR of the nearest interfering links in addition to the direct link; this will be referred to as local CSIR below. Under this assumption, we derive an exact analytical expression of the ergodic spectral efficiency attained by zero-forcing based successive interference cancellation (ZF-SIC) in terms of the relevant system parameters. By deriving a lower and an upper bound with closed forms on the sum spectral efficiency, we also demonstrate that the ergodic spectral efficiency scales with both the density of the links and the path-loss exponent, $\Theta(\lambda\log_{2}\left(1+\lambda^{\frac{\alpha}{2}(\beta-1)}\right))$ when $N_{\rm r} =c\lambda^{\beta}$ with some $c>0$ and $\alpha>2$.

\item We analyze the effects of receive antenna correlation and of a bounded path-loss function. An analytical expression of the lower bound on the sum spectral efficiency is derived as a function of the eigenvalues of a spatial correlation matrix when direct CSIR is known. A simple lower bound with a closed form  reveals that a linear scaling is still achievable with direct CSIR, provided the rank of the spatial correlation matrix scales in an appropriate super-linear way with the density. Furthermore, we find a sufficient condition for the number of receive antennas required to attain the linear scaling law with the direct CSIR when a bounded-path loss function is considered in the network. 
\end{itemize}

\begin{figure}
\centering
\includegraphics[width=5.5in]{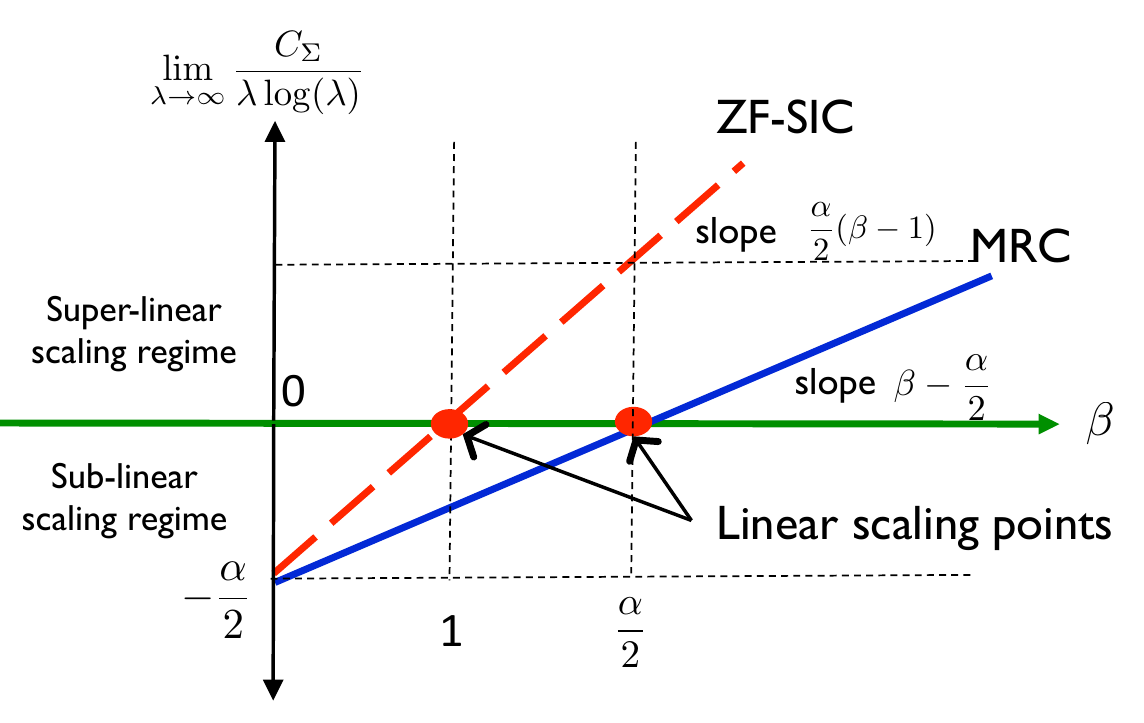} 
\caption{Asymptotic scaling behavior of the sum spectral spectral efficiency for ZF-SIC and MRC. When $\beta=1$ and $\beta=\frac{\alpha}{2}$, respectively, ZF-SIC and MRC achieve a linear growth of the sum spectral efficiency with respect to the density. If $\beta$ is less (resp. more) than the point that provides the linear growth, then the sum spectral efficiency increases sub-linearly (resp. decreases super-linearly) with the density. }
\label{fig:0}
\end{figure}

The exact expressions and scaling laws for the ergodic spectral efficiency are new findings. The capacity scaling result with the direct CSIR is partially aligned with the observation from a transmission capacity framework in \cite{Jindal_PZF,Huang,Hunter_MIMO}, where the linear scaling law of the transmission capacity is attained using MRC when the number of receive antennas scales in a certain super-linearly way. Our analysis confirms that this result holds from a sum spectral efficiency point-of-view and generalize to the case where the number of antenna scales with the density in a polynomial function with an arbitrary degree, i.e., $N_{\rm r}=c\lambda^{\beta}$, as illustrated in Fig. \ref{fig:0}. When the bounded path loss function is applied, we show that a  linear growth of the number of antennas is enough to maintain the linear capacity scaling, which is a new observation. Furthermore, our scaling result with local CSIR somewhat differs from the prior work  \cite{Jindal_PZF,Huang}, where the transmission capacity only scales with the density when the receivers can cancel interference from a set of nearest interferers while maximizing the desired signal power. Unlike this result, our analysis reveals that, when using local CSIR, the linear scaling law of the sum spectral efficiency is further improved with a multiplicative gain induced by the path loss exponent, as depicted in Fig. \ref{fig:0}. Further, it is shown that the antenna correlation degrades the sum spectral efficiency, especially when the condition number of the spatial correlation matrix is large, i.e., for highly correlated channels. Nevertheless,  linear scaling is still attainable with direct CSIR if the rank of the correlation matrix increases super-linearly with the density.

The paper is organized as follows. Section \ref{Section2} explains the network model and provides the performance metric. In Section \ref{Section3}, analytical expressions for the ergodic spectral efficiency are derived when only CSIR for the direct link is known. The case with local CSIR is analyzed in Section \ref{Section4}. Section \ref{Section5} provides analytical expressions for the sum spectral efficiency when antenna correlation and a bounded path loss function are incorporated. In Section \ref{Section6}, we provide conclusions and a discussion of future work.

%(nodes/m$^2$) where $0<\lambda<1$. For the case of $\lambda \geq 1$, with high probability, interferers exist in a near filed around a receiver, which is out of interest in many wireless systems.

\section{Model} \label{Section2}
In this section, we first describe network and signal models used in this paper. Then, we introduce the performance metrics.

\subsection{Network Model}  
 We consider a large random network where multiple transmit-and-receive pairs communicate in a common shared spectrum. We assume that the transmitters $\{ {\bf d}_k^{{\rm tx}}, k\in\mathbb{N} \}$ are distributed in the two-dimensional plane according to a homogeneous PPP $\Phi$ with density $\lambda$. The location at ${\bf d}_k^{\rm rx}$ of the receiver associated with the transmitter ${\bf d}_k^{{\rm tx}}$ is uniformly distributed 
 in the area of an annulus (ring) with inner radius 1 and outer radius $R_{\rm d}$, where $R_{\rm d}>1$. Here,  $R_{\rm d}$ determines the maximum communication range. Further, we assume that all transmissions are synchronous thanks to a common clock shared by the network. We assume all transmitters have a single antenna while each receiver is  equipped with $N_{\rm r}$ antennas. Our model differs from the ad hoc network models in \cite{Gupta,Leveque,Ozgur,Franceschetti} where source and destination pairs can be arbitrarily chosen. Rather, it is an extension of the bi-polar models used in \cite{Zhang_PC,Baccelli,Weber2,Weber_SIC, Jindal_SIC_JD,Haenggi_STG} by taking the random link distances within the fixed communication range $R_{\rm d}$ into account.

\begin{figure}
\centering
\includegraphics[width=5.5in]{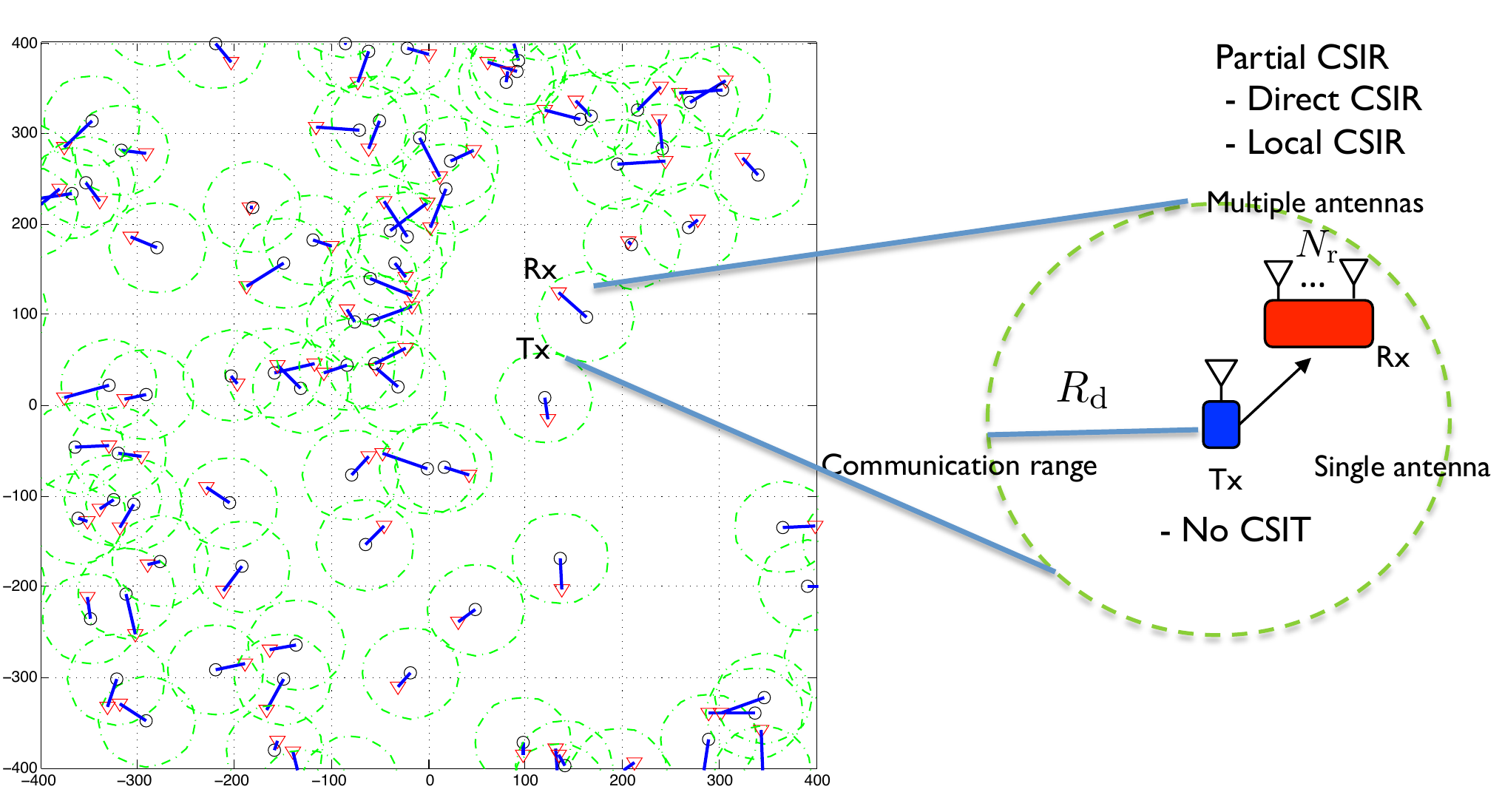} 
\caption{A snapshot of the network with density $\lambda=0.00005$. Each transmitter with a single antenna communciates with its associated receiver with $N_{\rm r}$ multiple receive antennas, which is  uniformly distributed in the area of an annulus with inner radius 1 and outer radius $R_{\rm d}$, where $R_{\rm d}>1$.}
\label{fig:sys}
\end{figure}

\subsection{Signal Model}
In a fixed area $\mathcal{A} \subset \mathbb{R}^2$, a random number $K$ of transmitters communicate by sharing the common spectrum; thereby each transmission interferes with each other. $K$ is a Poisson random variable with mean $\mathbb{E}[K] = \lambda|\mathcal{A}|$ where $\lambda$ can be interpreted as a spatial node intensity per unit area in the network and where $|\mathcal{A}|$ is the area of $\mathcal{A}$. Fig \ref{fig:sys} illustrates a snapshot of the network topology when $\lambda=0.00005$.

In a particular realization of $K$, transmitter $k\in\mathcal{K}$, where $\mathcal{K}=\{1,\ldots,K\}$, sends a message $W_{k}$ to its associated receiver. 
Let $s_{k}=f(W_{k})$ be the signal sent by transmitter $k$ where $f(\cdot)$ represents an encoding function such that the transmitted signal satisfies the power constraint $\mathbb{E}\left[|s_k|^2\right]\leq P$. ${\bf h}_{k,\ell}\in\mathbb{C}^{N_{\rm r}\times 1}$ and ${d}_{k,\ell}=\|{\bf d}_k^{\rm rx}-{\bf d}_{\ell}^{\rm tx}\|_2$ respectively represent the fading vector and the distance from the $k$th transmitter to the $j$th receiver. Further, $\alpha \in\mathbb{R}^{+}$ with $\alpha>2$ and ${\bf z}_k\in\mathbb{C}^{N_{\rm r}\times 1}$ respectively represent the path-loss exponent and the noise vector at the $k$th receiver.  Assuming a narrowband channel model, when all transmitters simultaneously send signals, the total received signal of the $k$th receiver, ${\bf y}_k\in\mathbb{C}^{N_{\rm r}\times 1}$, is given by
\begin{align}
{\bf y}_k \!= {\bf h}_{k,k}{d}_{k,k}^{-\frac{\alpha}{2}}s_k  +  \sum_{\ell \neq k }^K {\bf h}_{k,\ell} {d}_{k,\ell}^{-\frac{\alpha}{2}} s_{\ell} + {\bf z}_k. \label{eq:received_1}
\end{align}
We assume that all entries of ${\bf h}_{k,j}$ are independent and identically distributed (IID) complex Gaussian random variables each with zero mean and unit variance, i.e., $\mathcal{CN}(0,1)$. Furthermore, it is assumed that all entries of ${\bf z}_k$ are IID $\mathcal{CN}(0,\sigma^2)$, where $\sigma^2$ represents the variance of noise. 

%It is worth to note that unlike the prior work in \cite{Hunter_MIMO,Jindal_PZF, Huang}, we consider a bounded (non-singular) path loss function. This captures the different behaviors of the propagation loss in both near and far fields as shown in  \cite{Haenggi_pathloss_1,Franceschetti,Haenggi_pathloss_2,Schantz1,Schantz2} with strongly experimental support. 

\subsection{Sum Spectral Efficiency }
We define two achievable sum spectral efficiencies, each for a different CSIR assumption.

\subsubsection{Direct CSIR}
We first define an instantaneous signal-to-interference-plus-noise-ratio (SINR) when receiver $k\in\mathcal{K}$ exploits ${\bf h}_{k,k}$. This assumption is practically favorable because it requires the receiver to only learn  the direct link's channel, which can be done using a control channel with a reasonable amount of pilot signal overhead. With this CSIR, under the premise of no interference covariance matrix estimation, the optimal receiving strategy is to maximize the desired signal power using a MRC technique \cite{Tse}. Applying the MRC vector $
{\bf w}_{k}^{*}=\frac{{\bf h}_{k,k}^{*}}{ \|{\bf h}_{k,k}\|_2}$, where ${\bf x}^*$ is the complex conjugate of the transpose of vector ${\bf x}$, the instantaneous SINR of receiver $k$ is given by
\begin{align}
\textrm{SINR}^{{\rm mrc}}_k=\frac{H_{k,k}d_{k,k}^{-\alpha}}{ I_k+\frac{1}{{\rm SNR}}}, \label{eq:SINR_mrc}
\end{align}   
where $H_{k,k} =\frac{|{\bf h}_{k,k}^*{\bf h}_{k,k}|^2}{\|{\bf h}_{k,k}\|_2^2}= \|{\bf h}_{k,k}\|_2^2$ denotes the fading power of the direct link, distributed as a Chi-squared random variable with $2N_{\rm r}$ degrees of freedom. Further, ${\rm SNR}=\frac{P}{\sigma^2}$ and $I_k$ denotes the aggregated interference power:
\begin{align}
I_k= \sum_{j \in \mathcal{K}/\{k\}}H_{k,\ell}d_{k,\ell}^{-\alpha},
\end{align}
where $H_{k,\ell} =\frac{|{\bf h}_{k,k}^*{\bf h}_{k,\ell}|^2}{\|{\bf h}_{k,k}\|_2^2}$ represents the fading power of the interfering link from the $\ell$th transmitter to the $k$th receiver, which is an exponential random variable with mean one as shown in \cite{Jindal_PZF,Namyoon_SG}. Under the premise that each transmitter knows the effective SINR and uses adaptive modulation/coding to select the right rate, each link is able to achieve Shannon's bound for its instantaneous SINR, i.e., $\log_2(1+{\rm SINR}_k^{\rm mrc})$. Therefore, the sum of the spectral efficiencies per unit area is given by
\begin{align}
C_{\Sigma}^{{\rm mrc}} &=\frac{1}{\mathcal{A}}\mathbb{E}\left[\sum_{k=1}^K \log_2\left(1\!+\!\textrm{SINR}_k^{\rm mrc} \right)\right]\nonumber\\
&=\lambda\mathbb{E}^{o}\left[ \log_2\left(1\!+\!\textrm{SINR}_k^{\rm mrc} \right)\right], \label{eq:sumrate_mrc}
\end{align}
where $\mathbb{E}^o$ denotes the Palm probability of the PPP and the last equality follows from the definition of Palm probability \cite{Baccelli_book}. The expectations are taken over the multiple level of randomness associated with link distances and fadings. The analysis of this sum spectral efficiency will be presented in Section III.

\subsubsection{Local CSIR}
 We now consider a different assumption where each receiver uses channel knowledge of some limited number of interfering links in addition to that of its own link. Without loss of generality, we can order the interferers in increasing distance from receiver $k$ in such a way that $d_{k,k_1}<d_{k,k_2},\ldots, <d_{k,k_{K-1}}$, where ${d}_{k,k_j}=\|{\bf d}_k^{\rm rx}-{\bf d}_{k_j}^{\rm tx}\|_2$, for $k_j\in\mathcal{K}/\{k\}$. The inequalities are almost surely strict because, with probability 1, no two transmitters are at the same distance from the receiver. With the assumption that receiver $k$ knows CSIR for a certain set of nearest interfering links, we derive an instantaneous SINR expression when each receiver performs ZF-SIC \cite{VBLAST}. The idea of ZF-SIC decoding is to successively cancel the effects of neighbor interference signals before decoding the desired signal; thereby it provides both interference cancellation gain and a power gain in the SINR.  Under the premise that receiver $k$, for $k\in \mathcal{K}$, measures the $L$ nearest interferer channel vectors, i.e., $\{{\bf  h}_{k,k_{\ell}}\}$, for $k_{\ell}\in \mathcal{N}_{k}=\{k_1,\ldots, k_{L}\}$, where $L\leq N_{\rm r}-1$, it is able to construct a concatenated channel matrix ${\bf H}_k=\left[{\bf h}_{k,k},{\bf h}_{k,k_1},\ldots,{\bf h}_{k,k_{L}}\right] \in \mathbb{C}^{N_{\rm r}\times (1+L)}$. Applying the QR decomposition\cite{Caire}, the channel matrix ${\bf H}_k$ is a product of a unitary matrix ${\bf Q}_k \in \mathbb{C}^{N_{\rm r}\times N_{\rm r}} $ and an upper-triangular matrix ${\bf R}_k \in \mathbb{C}^{N_{\rm r}\times (1+L)}$, namely,
\begin{align}
{\bf H}_k={\bf Q}_k{\bf R}_k,
\end{align}
where $[{\bf R}_k]_{i,j}=0$ for $i>j$. Applying ${\bf Q}_k^*$ to the received signal vector in (\ref{eq:received_1}), the resulting input-output relationship is
\begin{align}
{\bf \tilde y}_k={\bf Q}_k^*{\bf y}_k 
&={\bf R}_{k} {\bf s}_k+\sum_{j\in \mathcal{K} \setminus \{\mathcal{N}_k \cup k\}} {\bf \tilde h}_{k,j}d_{k,j}^{-\frac{\alpha}{2}}s_j+{\bf \tilde z}_k,
\end{align}
where ${\bf s}_k=\left[d_{k,k}^{-\frac{\alpha}{2}}s_k,~d_{k,k_1}^{-\frac{\alpha}{2}}s_{k_1},\ldots,~d_{k,k_{L}}^{-\frac{\alpha}{2}}s_{k_{L}}\right]^T$, ${\bf \tilde h}_{k,j}={\bf Q}_k^*{\bf h}_{k,j}$, and ${\bf \tilde z}_{k}={\bf Q}_k^*{\bf z}_{k}$. Since ${\bf Q}_k$ is a unitary matrix and the channel is IID complex Gaussian, the distribution of ${\bf \tilde h}_{k,j}$ (resp. ${\bf \tilde z}_{k}$) is the same as that of ${\bf  h}_{k,j}$ (resp. ${\bf  z}_{k}$). 

Assuming that successive interference cancellation is used, under the premise that each receiver knows the modulation and coding methods of the nearest interfering transmitters, all data streams sent by the $L$ nearest interferers are decoded and can thus be subtracted from the first element of ${\bf \tilde y}_k$, i.e., ${\bf \tilde y}_k(1)$. After subtracting the nearby interferer contributions, we have the following equivalent input-output relationship for decoding the $s_{k}$ data stream:
\begin{align}
{\tilde {y}}_{k}= \tilde{h}_{1,1}d_{k,k}^{-\frac{\alpha}{2}}s_{k} +\sum_{j\in \mathcal{K}\setminus \{\mathcal{N}_k\cup k\}} { \tilde h}_{k,j}d_{k,j}^{-\frac{\alpha}{2}}s_j+{\tilde z}_k, \label{eq:input_output}
\end{align}
where ${\tilde {y}}_{k}={\tilde {\bf y}}_{k}(1)$, ${\tilde h}_{1,1}={\bf R}_k(1,1)$, $\tilde{h}_{k,j}={\bf \tilde h}_{k,j}(1)$, and $\tilde{z}_k={\bf \tilde z}_{k}(1)$. Consequently, the resulting instantaneous SINR of receiver $k$ is given by
\begin{align}
\textrm{SINR}^{{\rm sic}}_k=\frac{\tilde{H}_{k,k}d_{k,k}^{-\alpha}}{ \tilde{I}_k+\frac{1}{{\rm SNR}}}, \label{eq:SINR_sic}
\end{align}   
where $\tilde{H}_{k,j}=|\tilde{h}_{k,j}|^2$ is distributed as a Chi-squared random variable with $2N_{\rm r}$ degrees of freedom \cite{Caire}, ${\rm SNR}=\frac{P}{\sigma^2}$, and $\tilde{I}_k$ denotes the aggregated interference power
\begin{align}
\tilde{I}_k= \!\!\!\!\!\sum_{j \in \mathcal{K}\setminus\{\mathcal{N}_k\cup{k}\}}\!\!\!\!\!\!\tilde{H}_{k,j}d_{k,j}^{-\alpha},
\end{align}
where $\tilde{H}_{k,j}=|\tilde{h}_{k,j}|^2$ is an exponential random variable with mean one as shown in \cite{Caire}.

Consequently, the sum of spectral efficiencies per unit area achieved by the ZF-SIC is given by
\begin{align}
C_{\Sigma}^{{\rm sic}} &=\lambda\mathbb{E}^o\left[ \log_2\left(1+\textrm{SINR}_k^{\rm sic}\right)\right]. \label{eq:sumrate_sic}
\end{align}
The analysis for the sum of spectral efficiencies with this local CSIR will be given in Section V.

It is worthwhile to mention that the sum spectral efficiencies in (\ref{eq:sumrate_mrc}) and (\ref{eq:sumrate_sic}) are the result of averaging over 1) all fading distributions depending on the receiving strategies and 2) all realizations of the network topology under the Poisson assumption. %. Therefore, it provides results that are valid for general classes of networks, rather than being restricted to a particular deployment of the nodes in a network. 

In this paper, we use the following asymptotic notation \cite{Knuth}. 1) $f(\lambda)=O(g(\lambda))$ if $f(\lambda)\leq kg(\lambda)$ as $\lambda$ tends to infinity for some constant $k$, 2) $f(\lambda)=\Theta(g(\lambda))$ if $k_1g(\lambda)\leq f(\lambda)\leq k_2g(\lambda)$ as $\lambda$ tends to infinity for some constants $k_1$ and $k_2$, 3) $f(\lambda)=\Omega(g(\lambda))$ if $f(\lambda)\geq kg(\lambda)$ as $\lambda$ tends to infinity for some constant $k$.

%%%%%%%%%%%%%%%%%%%%%%%%%%
%%%%%%%%%%%%%%%%%%%%%%%%%%%
 
\section{Direct CSIR} \label{Section3}
In this section, we analyze the ergodic spectral efficiency and the scaling behavior of the network described in Section II when the receiver only exploits CSIR for the direct link. We first provide an exact characterization of the sum spectral efficiency and then derive the scaling law.

\subsection{Analytical Characterization}

The analytical characterization relies on a lemma introduced in \cite{Hamdi}. This Lemma provides in integral expression of the ergodic spectral efficiency as a function of the Laplace transforms of both the desired signal power and the aggregated interference power. For the sake of completeness, we reproduce it below.

\begin{lemma} \label{lem1} Let $X>0$ and $Y>0$ be non-negative and independent random variables. Then, for any $a>0$,
\begin{align}
&\mathbb{E}\left[\ln\left(1+\frac{X}{Y+a}\right)\right] =\int_{0}^{\infty}\frac{e^{-az}}{z}\left(1-\mathbb{E}\left[e^{-zX}\right]\right)\mathbb{E}\left[e^{-zY}\right] \d z.
\end{align}
  \end{lemma}
\begin{proof} See \cite{Hamdi}.
\end{proof}

Using Lemma \ref{lem1}, we present our main result for the ergodic spectral efficiency in an integral form.

\begin{theorem} \label{Th1} The sum of spectral efficiencies with direct CSIR is 
\begin{align}
C_{\Sigma}^{{\rm mrc}}  
=\frac{\alpha}{2\ln(2)} \int_{1}^{R_{\rm d}} \int_{0}^{\infty}\frac{e^{-\frac{\left( \sinc\left(\frac{2}{\alpha}\right)u    \right)^{\! \frac{\alpha}{2}}  }{(\lambda \pi)^{\frac{\alpha}{2}}{\rm SNR}}-u}}{u}  \frac{\sum_{n=1}^{N_{\rm r}} \binom{N_{\rm{r}}}{n} \left(\frac{\sinc\left(\frac{2}{\alpha}\right)u }{\lambda \pi r^{2}}\right)^{\!\!\! n\frac{\alpha}{2}}   }{\left(1+\left(\frac{\sinc\left(\frac{2}{\alpha}\right)u }{\lambda \pi r^{2}}\right)^{\!\!\! \frac{\alpha}{2}}  \right)^{N_{\rm r}}}  \d u \frac{2r}{R_{\rm d}^2-1}\d r. \label{eq:Th1}
\end{align} 
\end{theorem}
\begin{proof}
See Appendix \ref{proof:Th1}.
\end{proof}

\begin{figure}
\centering
\includegraphics[width=4.5in]{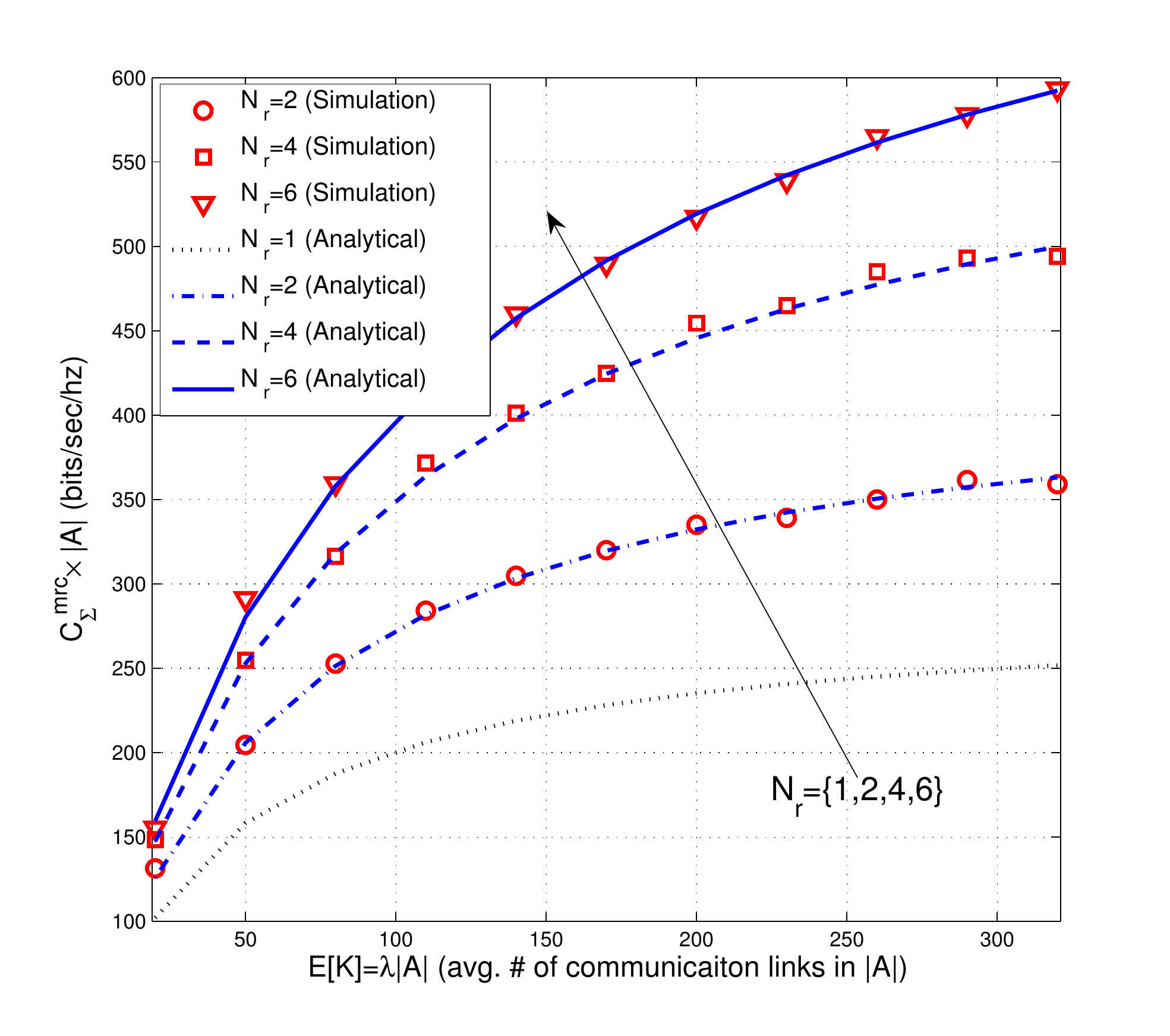} 
\caption{The sum of spectral efficiencies achieved with direct CSIR when $|\mathcal{A}|=\pi 500^2$ ($m^2$), $\alpha=4$, $R_{\rm d}=50$m, $P=20$ (dBm), and $\sigma^2=-104$ (dBm).}
\label{fig:1}
\end{figure}

The sum spectral efficiency depends on the relevant system parameters, chiefly the density of links $\lambda$, the number of antennas at the receiver $N_{\rm r}$, the path-loss exponent $\alpha$, the communication range $R_{\rm d}$, and the operating SNR. 
This formula generalizes the expression given in \cite{Hunter_MIMO} in the sense that it incorporates both the randomness on the direct link's distance and noise effects, avoiding sophisticated differentiations of the Laplace transform of the interference power. To back Theorem \ref{Th1}, we compare (\ref{eq:Th1}) with simulation results in Fig. \ref{fig:1}, when $\alpha=4$, ${\rm SNR}=84$ (dB), and $R_{\rm d}= 50$ m. The agreement is excellent for the various values of $N_{\rm r}$ considered and for the entire range of $\lambda$ of interest. As $N_{\rm r}$ increases, the sum of the spectral efficiencies improves because the desired signal power is boosted by the array gain. In particular, the gain is significant in the regime of low $\lambda$ and  saturates as the density increases. This fact reveals that MRC is simple yet effective to improve the sum spectral efficiency when the density of link is low. This is because, in the low density regime, boosting the signal power while treating interference as noise is asymptotically optimal \cite{TIN_paper}.

To provide a more transparent interpretation of the expression in Theorem \ref{Th1}, we consider the following examples.

{\bf Example 1:} The simplest scenario is that where $N_{\rm r}=1$, $\alpha=4$, and $d_{k,k}=d$. In this case, in the interference limited regime, Theorem \ref{Th1} gives 
\begin{align}
C_{\Sigma}^{{\rm mrc}}=\frac{2\lambda}{\ln(2)} \left\{\sin\left(\frac{\lambda \pi^2  d^2}{2}\right)\left(\frac{\pi}{2}-{\rm Si}\left(\frac{\lambda \pi^2  d^2}{2}\right) \right)-\cos\left(\frac{\lambda \pi^2  d^2}{2}\right){\rm Ci}\left(\frac{\lambda \pi^2  d^2}{2}\right)\right\}, \label{eq:single_antenna_mrc}
\end{align} 
where ${\rm Si}(z)=\int_{0}^{z}\frac{\sin(t)}{t}\d t$ and ${\rm Ci}(z)=-\int_{z}^{\infty}\frac{\cos(t)}{t}\d t$ respectively represent  the cosine and the sine integral function. This analytical expression is useful to understand the interplay between the link distance and the density in the capacity scaling law. For example, if we shrink the link distance $d=\frac{1}{\sqrt{\lambda \pi}}$, the sum spectral efficiency boils down to
\begin{align}
C_{\Sigma}^{{\rm mrc}}=\frac{2\lambda}{\ln(2)}\left(\frac{\pi}{2}-{\rm Si}\left(\frac{  \pi}{2}\right)\right)\simeq 0.5772 \lambda.
\end{align} 
This example shows that it is possible to obtain a linear growth of the capacity with the density, i.e., $\Theta(\lambda)$, provided the link distance scales down as $d=\Theta\left(\frac{1}{\sqrt{\lambda}}\right)$ when $N_{\rm r}=1$.

{\bf Example 2:} For the given link distance $d_{k,k}=d$, in interference limited networks, Theorem \ref{Th1} simplifies further to a single integral form as 
\begin{align}
C_{\Sigma}^{{\rm mrc}}=\frac{\lambda\alpha}{2\ln(2)}\sum_{n=1}^{N_{\rm r}} \binom{N_{\rm{r}}}{n} \int_{0}^{\infty}\frac{e^{-u}}{u}  \frac{ \left(\frac{\sinc\left(\frac{2}{\alpha}\right)u }{\lambda \pi d^{2}}\right)^{\!\!\! n\frac{\alpha}{2}}   }{\left(1+\left(\frac{\sinc\left(\frac{2}{\alpha}\right)u }{\lambda \pi d^{2}}\right)^{\!\!\! \frac{\alpha}{2}}  \right)^{N_{\rm r}}}  \d u.
\end{align} 
This expression provides a better understanding of the sum spectral efficiency performance than the expression given in Theorem \ref{Th1}. For instance, it is possible to observe that the sum spectral efficiency per link increases with the number of receive antennas $N_{\rm r}$.

{\bf Example 3:} For a given set of network parameters, the MRC technique that treats all interference as noise could be an optimal receiving strategy for a certain fraction of all communication links. Since all communication links experience the same SNR and interference-to-noise ratio (INR) distributions under the PPP to compute its fraction, we compute the Palm probability that a typical receiver satisfies the condition for being scheduled by the ITLinQ scheduling algorithm in \cite{ITLinQ}. In particular, to make this computation simple yet capturing the interplay between the density and the number of receive antennas, we use channel hardening assumptions \cite{Hochwald}, which essentially hold when a large number of receive antennas is used.

Suppose that, when using MRC, the $k$th link obtains a deterministic array gain $N_{\rm r}$, and that the fading power from the interferer is a constant and equals 1, for $\ell\neq k$. Conditioned on $d_{k,k}=d$, the probability that the typical receiver satisfies the ITLinQ (destination) condition in \cite{ITLinQ} is,
\begin{align}
\mathbb{P}\left[\sqrt{{\rm SNR}_{k,k}}\geq {\rm INR}_{k,k_1}\right]
&=
\mathbb{P}\left[\sqrt{\frac{N_{\rm r}P}{d^{\alpha}\sigma^2}}\geq  \frac{P}{d_{k,k_1}^{\alpha}\sigma^2}\right]\nonumber \\
&=\mathbb{P}\left[d_{k,k_1}\geq  \left(\frac{P}{\sigma^2 N_{\rm r}}\right)^{\!\!\frac{1}{2\alpha}}  d^{\frac{1}{2}}\right]
\nonumber \\
&\stackrel{(a)}{=} \mathbb{P}\left[ \Phi\left(\mathcal{B}\left(0, \left(\frac{P}{\sigma^2 N_{\rm r}}\right)^{\!\!\frac{1}{2\alpha}}  d^{\frac{1}{2}}\right)=0\right) \right]  
\nonumber \\
&= \exp\left(-\lambda \pi  \left(\frac{P}{\sigma^2 N_{\rm r}}\right)^{\!\!\frac{1}{\alpha}}  d  \right),
%\nonumber \\
%&=\frac{2\left(1-e^{-\lambda \pi  \left(\frac{P}{\sigma^2 N_{\rm r}}\right)^{\!\!\frac{1}{\alpha}}  R_{\rm d}} (1)\right)}{(\lambda \pi R_{\rm d})^2  \left(\frac{P}{\sigma^2 N_{\rm r}}\right)^{\!\!\frac{2}{\alpha}}}
\end{align} 
where (a) comes from the fact that the probability that the distance to the nearest interferer is greater than $x>0$ equals the probability that there is no interferer in the ball with radius $x$. This expression shows the benefits of MRC. For the given density, the probability that the optimality condition of treating interference noise is satisfied increases as the number of antennas increases. In an asymptotic sense, if we scale up the number of antennas with the density in such a way that $\lim_{\lambda\rightarrow \infty}\frac{\lambda}{N_{\rm r}^{\frac{1}{\alpha}}}=0$, then applying MRC while treating all interference signals as noise is the optimal strategy with high probability. 

%This example, however, should be carefully interpreted as it ignored the random effects from the channel fadings and the link distance. 

\subsection{Scaling Law}
Although general and exact, the expression given in Theorem \ref{Th1} is rather complicated, propelling the interest in more compact characterizations. Still in full generality, in this section, we provide a lower and upper bound with a closed-form to the sum spectral efficiency, which allows to prove the scaling law. The lower and upper bounds are derived using both Lemma \ref{lem1} and the following lemma, which uses the first order moment measure of the signal power and the interference power to establish the bounds. 

\begin{lemma} \label{lem2} Let $X>0$ and $Y>0$ be independent non-negative random variables such that $\mathbb{E}[X] <\infty$, $\mathbb{E}[\ln(X)] <\infty$, $\mathbb{E}\left[\frac{1}{Y}\right]<\infty$, and $\mathbb{E}[Y] <\infty$. Then, 
\begin{align}
\log_2\left(1+\frac{\exp\left(\mathbb{E}[\ln(X)]\right)}{\mathbb{E}\left[Y\right]}\right)\leq \mathbb{E}_{X,Y}\left[\log_2\left(1+\frac{X}{Y}\right)\right]\leq 
\log_2\left(1+ \mathbb{E}[X]\mathbb{E}\left[\frac{1}{Y}\right]\right).
\end{align}
\end{lemma}
\begin{proof}
See Appendix \ref{proof:lem2}.
\end{proof}

Leveraging Lemma \ref{lem2}, we provide the linear scaling law in networks with direct CSIR in the following theorem. 
 
 \begin{theorem}[Scaling law with direct CSIR] \label{Th2} Assume that $N_{\rm r}=c\lambda^{\beta}$ for some $c>0$ and $\beta\geq 0$. Then, in the interference limited regime ($\sigma^2=0$), the ergodic spectral efficiency of a typical link scales with the density as follows:
\begin{align}
\frac{C_{\Sigma}^{{\rm mrc}}}{\lambda} =\Theta\left(  \log_2\left(1+\lambda^{\beta-\frac{\alpha}{2}}\right)\right),
\end{align}
as $\lambda \rightarrow \infty$.
\end{theorem}

%and 2) $\lim_{\lambda \rightarrow \infty} \frac{C_{\Sigma}^{{\rm mrc}}}{\lambda} = 0$ when $\beta < \frac{\alpha}{2}$\begin{align}
%\lim_{\lambda \rightarrow \infty} \frac{C_{\Sigma}^{{\rm mrc}}}{\lambda} = \Theta(1).
%\end{align}$\lim_{\lambda \rightarrow \infty} \frac{C_{\Sigma}^{{\rm sic}}}{\lambda} = \Omega\left(  \log_2\left(1+\lambda^{\frac{\alpha}{2}(\beta-1)}\right)\right)$ by deriving a lower bound with a closed from. \left\{
%%\begin{array}{l l}
%  0 & \textrm{when $\beta<\frac{\alpha}{2}$} \\
 % \Theta(1)  & \textrm{when $\beta=\frac{\alpha}{2}$} \\
 % \Theta\left((\beta-\frac{\alpha}{2})\log(\lambda)\right)  & \textrm{when $\beta > \frac{\alpha}{2}$}.
%\end{array} \right.

 \begin{proof}
 We begin with establishing a lower bound on $C_{\Sigma}^{\rm mrc}$ to prove that $\frac{C_{\Sigma}^{{\rm mrc}}}{\lambda}=\Omega\left(  \log_2\left(1+\lambda^{\beta-\frac{\alpha}{2}}\right)\right)$. From Lemma \ref{lem2}, in the interference limited regime, the sum spectral efficiency is lower bounded as follows:
\begin{align}
&\lambda\mathbb{E}_{H_{k,k},d_{k,k},I_k}\left[\log_2\!\left(1+\frac{H_{k,k}d_{k,k}^{-\alpha}}{I_k +\frac{1}{\rm SNR}}\right) \right] \geq\lambda \mathbb{E}_{d_{k,k},I_k}\left[\log_2\left(1\!+\!\frac{ e^{\mathbb{E}\left[\ln(H_{k,k})\right] } }{   d^{\alpha}_{k,k}I_k  }\right)\right].\label{eq:mrc_lower_2}
%\nonumber \\
%&=\lambda \mathbb{E}_{I_k}\left[\log_2\left(1\!+\!\frac{ \exp\left(\mathbb{E}\left[\ln(H_{k,k})\right] \right) }{   \frac{2}{2+\alpha}R^{\alpha}_{{\rm d} }I_k  }\right)\right]
\end{align}
Using the fact that $H_{k,k}$ is a Chi-squared random variable with $2N_{\rm r}$ degrees of freedom, we obtain
\begin{align}
\mathbb{E}\left[\ln(H_{k,k})\right]= \ln(2)+\psi\left(\frac{N_{\rm r}}{2}\right) ,\label{eq:mean_sig_mrc}\end{align}
where 
\begin{align}
\psi(n)=-\gamma +\sum_{q=1}^{n-1}\frac{1}{q},
\end{align}
with $\gamma=\int_{0}^{\infty}\ln(x)e^{-x}\d x=0.57721566$ the Euler-Mascheroni's constant.
Using the inequality $e^{\ln(2)+\psi\left(\frac{N_{\rm r}}{2}\right)}\geq N_{\rm r}-1$ for all $N_{\rm r}> 1$, we obtain
\begin{align} 
e^{\mathbb{E}\left[\ln(H_{k,k})\right]} \geq N_{\rm r}-1.
\end{align}
We now use Lemma \ref{lem1} to obtain a lower bound. With Lemma \ref{lem1}, the lower bound in (\ref{eq:mrc_lower_2}) can be rewritten in an integral form as
\begin{align}
\lambda \mathbb{E}_{d_{k,k},I_k}\left[\log_2\left(1\!+\!\frac{ N_{\rm r}-1}{ d_{k,k}^{\alpha} I_k  }\right)\right]&=\frac{\lambda}{\ln(2)}\int_{0}^{\infty}\frac{1}{z}\left(1-e^{-z (N_{\rm r}-1)}\right)\mathbb{E}_{d_{k,k},I_{k,k}}\left[e^{-zd_{k,k}^{\alpha} I_k}\right]\d z \nonumber \\
&\stackrel{(a)}{=}\frac{\lambda}{\ln(2)}\int_{0}^{\infty}\frac{1}{z}\left(1-e^{-z (N_{\rm r}-1)}\right)\mathbb{E}_{d_{k,k}}\left[e^{-\frac{\lambda\pi d_{k,k}^2}{\sinc\left(\frac{2}{\alpha}\right)}z^{\frac{2}{\alpha}}}\right]\d z \nonumber \\
&\stackrel{(b)}{\geq}\frac{\lambda}{\ln(2)}\int_{0}^{\infty}\frac{1}{z}\left(1-e^{-z (N_{\rm r}-1)}\right)e^{-\frac{\lambda\pi \mathbb{E}\left[d_{k,k}^2\right]}{\sinc\left(\frac{2}{\alpha}\right)}z^{\frac{2}{\alpha}}}\d z \nonumber \\
&\stackrel{(c)}{=}\frac{\alpha\lambda}{2\ln(2)}\int_{0}^{\infty}\frac{e^{-u}}{u}\left\{1-e^{-  \left(2\sinc\left(\frac{2}{\alpha}\right)\right)^{ \frac{\alpha}{2}} \frac{ N_{\rm r}-1 }{ \left(\lambda \pi \left(R_{\rm d}^2+1\right) \right)^{\frac{\alpha}{2}}}  u^{\frac{\alpha}{2}}  }  \right\}\d u
\nonumber \\
&\stackrel{(d)}{\geq}\frac{\lambda}{\ln(2)}\int_{0}^{\infty}\frac{e^{-u^{\frac{\alpha}{2}}  }}{u}\left\{1-e^{-  \left(2\sinc\left(\frac{2}{\alpha}\right)\right)^{ \frac{\alpha}{2}} \frac{ N_{\rm r}-1 }{ \left(\lambda \pi \left(R_{\rm d}^2+1\right)\right)^{\frac{\alpha}{2}}}  u^{\frac{\alpha}{2}} }  \right\}\d u, \nonumber \\
%&+\frac{\alpha\lambda}{2\ln(2)}\int_{1}^{\infty}\frac{e^{-u^{\frac{\alpha}{2}}  }  }{u}\left\{1-e^{-4\sinc\left(\frac{2}{\alpha}\right)^{ \frac{\alpha}{2}} \frac{ N_{\rm r}-1 }{ \left(\lambda \pi R_{\rm d}^2 \right)^{\frac{\alpha}{2}}}  u^{\frac{\alpha}{2}}  }  \right\}\d u, \nonumber \\
%&=E \nonumber \\
&= \lambda\frac{2}{\alpha} \log_2\left(1+ \left(2\sinc\left(\frac{2}{\alpha}\right)\right)^{ \frac{\alpha}{2}} \frac{ N_{\rm r}-1 }{ \left(\lambda \pi \left(R_{\rm d}^2+1\right)\right)^{\frac{\alpha}{2}}} \right), \label{eq:lower_integ_1}
\end{align}
where (a) follows from the expression for the Laplace functional of the PPP, (b) follows from Jensen's inequality and $\mathbb{E}\left[d_{k,k}^{2}\right]=\frac{R_{\rm d}^2+1}{2}$, (c) comes from the variable change in (\ref{eq:condi_ergodic_rate_2}), and (d) follows from the fact that $e^{-u}\geq \frac{2}{\alpha}e^{-u^{\frac{\alpha}{2}}}$ for $u\geq 0$ when $\alpha >2$. Using the assumption that $N_{\rm r}=c\lambda^{\beta}$, as $\lambda$ goes to infinity for the given $R_{\rm d}$, we obtain
\begin{align}
\lim_{\lambda \rightarrow \infty}\frac{C_{\Sigma}^{{\rm mrc}}}{\lambda} &\geq  \frac{2}{\alpha}\log_2\left(1+ \frac{   \left(2\sinc\left(\frac{2}{\alpha}\right)\right)^{ \frac{\alpha}{2}} }{  \left( \pi \left(R_{\rm d}^2+1\right)\right)^{\frac{\alpha}{2}} } c\lambda^{\beta-\frac{\alpha}{2}}\right), \label{eq:lower_integ_22}
\end{align}
with $c>0$.
%As a result, the sum spectral efficiency at least linearly increases with the density, i.e., $C_{\Sigma}^{\rm mrc}=\Omega(\lambda)$. It is also possible to show that if $N_{\rm r}=c\lambda^{\beta}$ with $\beta<\frac{\alpha}{2}$, then $\lim_{\lambda \rightarrow \infty}\frac{C_{\Sigma}^{{\rm mrc}}}{\lambda}\geq 0$.

Next, we derive an upper bound on $C_{\Sigma}^{\rm mrc}$ to show $\frac{C_{\Sigma}^{{\rm mrc}}}{\lambda}=O\left(  \log_2\left(1+\lambda^{\beta-\frac{\alpha}{2}}\right)\right)$.
From Lemma \ref{lem2}, in the interference limited regime, the sum spectral efficiency is upper bounded as follows:
\begin{align}
\lambda\mathbb{E}_{H_{k,k},d_{k,k},I_k}\left[\log_2\!\left(1+\frac{H_{k,k}d_{k,k}^{-\alpha}}{I_k}\right) \right] &\leq\lambda \log_2\left(1\!+\mathbb{E}\left[d_{k,k}^{-\alpha}\right]\mathbb{E}[H_{k,k}] \mathbb{E}\left[\frac{1}{ I_k}\right] \right)\nonumber\\
&=\lambda  \log_2\left(1\!+  \frac{2\left(1-R_{\rm d}^{2-\alpha}\right)}{(\alpha-2)\left(R_{\rm d}^2-1\right)} N_{\rm r} \mathbb{E}\left[\frac{1}{ I_k}\right]\right),
\label{eq:upper_1}
\end{align}
where the last equality is due to the facts that $\mathbb{E}[H_{k,k}]=N_{\rm r}$ and $\mathbb{E}\left[d_{k,k}^{-\alpha}\right]=\frac{2\left(1-R_{\rm d}^{2-\alpha}\right)}{(\alpha-2)\left(R_{\rm d}^2-1\right)}$. To this end, we only need to compute a negative moment of the aggregated interference power. The negative moment is computed as follows:
\begin{align}
\mathbb{E}\left[\frac{1}{ I_k}\right]
&\stackrel{(a)}{=}\int_{0}^{\infty}\mathbb{E}\left[ e^{-uI_k}\right] \d u
\nonumber\\
&=\int_{0}^{\infty}e^{-\frac{\lambda\pi}{\sinc\left(\frac{2}{\alpha}\right)}u^{\frac{2}{\alpha}}} \d u
\nonumber\\
&=\frac{\Gamma\left(1+\frac{\alpha}{2}\right) \sinc\left(\frac{2}{\alpha}\right)^{\frac{\alpha}{2}}}{(\lambda \pi)^{\frac{\alpha}{2}}} , \label{eq:inter_upper}
\end{align}
where (a) follows from the relation $\mathbb{E}\left[\frac{1}{X}\right]=\mathbb{E}\left[\int_{0}^{\infty}e^{-sX} \d s\right]$ for any positive random variable $X$. Invoking (\ref{eq:inter_upper}) into (\ref{eq:upper_1}), the upper bound is given as follows:
\begin{align}
\lambda  \log_2\left(1\!+  \frac{2\left(1-R_{\rm d}^{2-\alpha}\right)}{(\alpha-2)\left(R_{\rm d}^2-1\right)} N_{\rm r} \mathbb{E}\left[\frac{1}{ I_k}\right]\right) &\leq \lambda  \log_2\left(1\!+  \frac{2\left(1\!-\!R_{\rm d}^{2-\alpha}\right)\Gamma\left(1\!+\!\frac{\alpha}{2}\right) \sinc\left(\frac{2}{\alpha}\right)^{\!\!\frac{\alpha}{2}}}{(\alpha-2)\left(R_{\rm d}^2-1\right)}  \frac{N_{\rm r}}{(\lambda \pi )^{\frac{\alpha}{2}} } \right) \nonumber \\
&=  \lambda  \log_2\!\left(\!1\!+  \frac{2\left(R_{\rm d}^{\alpha}\!-\!R_{\rm d}^{2}\right)\Gamma\left(1\!+\!\frac{\alpha}{2}\right) \sinc\left(\frac{2}{\alpha}\right)^{\!\!\frac{\alpha}{2}}}{(\alpha-2)\left(R_{\rm d}^2-1\right)}  \frac{N_{\rm r}}{(\lambda \pi R_{\rm d}^2 )^{\frac{\alpha}{2}} }\!\!\right).
\end{align} 
As a result, since $N_{\rm r}=c\lambda^{\beta}$, \begin{align}
\lim_{\lambda \rightarrow \infty}\frac{C_{\Sigma}^{{\rm mrc}} }{\lambda}\leq \log_2\left(1\!+ \frac{2\left(R_{\rm d}^{\alpha}\!-\!R_{\rm d}^{2}\right)\Gamma\left(1+\frac{\alpha}{2}\right) \sinc\left(\frac{2}{\alpha}\right)^{\frac{\alpha}{2}}}{(\alpha-2)\left(R_{\rm d}^2-1\right)  \left( \pi R_{\rm d}^2 \right)^{\frac{\alpha}{2}}}c\lambda^{\beta-\frac{\alpha}{2}}\!\right),  \label{eq:lower_integ_23}
\end{align}
with $c>0$. This completes the proof.
\end{proof}

This scaling result implies that there exists a critical  scaling of the number of receiver antennas to obtain a linear growth of $C_{\Sigma}^{\rm mrc}$, namely, $C_{\Sigma}^{\rm mrc}=\Theta(\lambda)$. To obtain a linear growth of $C_{\Sigma}^{\rm mrc}$ as the node density $\lambda$ increases, the number of receive antennas should be super-linearly scaled up with the density like $\lambda^{\frac{\alpha}{2}}$. This result backs the intuition that the receiver should boost the desired signal power more rapidly than the density to keep a constant SIR when MRC is applied. Meanwhile, for any $\beta$ with $\beta<\frac{\alpha}{2}$, the sum spectral efficiency asymptotically approaches zero because the SIR keeps decreasing as the density increases. If we scale up the number of receive antennas like $N_{\rm r}=\lambda^{\beta}$ where $\beta>\frac{\alpha}{2}$, then the sum spectral efficiency increases super-linearly with the density, i.e., like $\Theta(\lambda(\beta-\frac{\alpha}{2}) \log(\lambda))$.

One potential concern for this scaling result with the density is that the farfield assumption on the path loss model eventually does not hold as the density goes to infinity. This concern can be resolved by equivalently interpreting our scaling result in terms of the average number of interferers in the communication area, i.e., $\lambda \pi (R_{\rm d}^2-1)$. When the density is small enough to guarantee the farfield assumption with probability one (e.g., $\lambda=0.00005$), it is possible to increase the communication range $R_{\rm d}$ asymptotically, i.e., the average number of interferers goes to infinity. Then, to maintain the constant transmission rate as the average number of interferers increases, the number of receive antennas should be super-linearly scaled up with the average number of interfering transmitters, $\lambda \pi R_{\rm d}^2$ in a particular way, i.e., $N_{\rm r}=(\lambda \pi R_{\rm d}^2)^{\frac{\alpha}{2}}$. Although this interpretation could be helpful to understand the merits of using multiple receive antennas in an engineering sense, we shall characterize the capacity scaling of the network in terms of the density for the mathematical connivence in the rest of this paper.

{\bf Example 4:} When the number of receive antennas does not scale with the density, i.e., $\beta=0$, the scaling law per link boils down to $\Theta\left(\log_2\left(1+\frac{1}{\lambda^{\frac{\alpha}{2}}} \right)\right)\simeq \Theta\left(\lambda^{-\frac{\alpha}{2}}\right)$. This implies that the typical user's transmission rate goes down super-linearly with the density, and the lesser path-loss exponent causes the more transmission rate degradation. It is worthwhile to mention that this scaling result is more pessimistic than the well-known ad hoc capacity scaling law, $\Theta\left(\frac{1}{\sqrt{\lambda}}\right)$, in \cite{Gupta}. The discrepancy inherently follows from the different assumptions used in the two network models. In our model, the link association is fixed, and there is a non-zero probability that the nearest interferer's location can be arbitrary close to the typical link's receiver. Whereas, in the ad hoc network model \cite{Gupta}, the source and destination paris are randomly determined, and the transmission rate per link does not depend on the density due to the interference guard region, when the nearest neighbor routing algorithm is applied. In the ad hoc model \cite{Gupta}, instead of the transmission rate per link, the capacity scaling is crucially determined by the number of hops in a typical communication pair that is of order $\Theta\left(\frac{1}{\sqrt{\lambda}}\right)$. Whereas, in our model, the capacity scaling is decided by the transmission rate per link due to the single-hop communication constraint.

{\bf Example 5:} When $\alpha=4$ and $d_{k,k}=d$, in the interference limited regime, the proofs of Theorems \ref{Th1} and \ref{Th2} show that the sum spectral efficiency can be approximated as  
\begin{align}
C_{\Sigma}^{{\rm mrc}}\simeq \frac{2\lambda}{\ln(2)} \left\{\sin\left(\frac{\lambda \pi^2  d^2}{2\sqrt{N_{\rm r}}}\right)\left(\frac{\pi}{2}-{\rm Si}\left(\frac{\lambda \pi^2  d^2}{2\sqrt{N_{\rm r}}}\right) \right)-\cos\left(\frac{\lambda \pi^2  d^2}{2\sqrt{N_{\rm r}}}\right){\rm Ci}\left(\frac{\lambda \pi^2  d^2}{2\sqrt{N_{\rm r}}}\right)\right\}. \label{eq:m_antenna_mrc}
\end{align} 
As shown in Theorem \ref{Th2}, if we scale up the number of receive antennas with the density as $N_{\rm r}=(\lambda \pi d^2)^2$, the sum spectral efficiency is simply given by
\begin{align}
C_{\Sigma}^{{\rm mrc}}\simeq \frac{2\lambda}{\ln(2)}\left(\frac{\pi}{2}-{\rm Si}\left(\frac{  \pi}{2}\right)\right)\simeq 0.5772 \lambda.
\end{align} 
Note that this is the same expression shown in Example 1. Therefore, the role of MRC can be interpreted as virtually reducing the link distance by boosting the direct channel gain.

{\bf Example 6:}
In the network, one interesting question would be to determine the link density $\lambda$ for a given set of system parameters, which maximizes the sum of spectral efficiencies. For this, one can leverage the lower bound on the sum spectral efficiency in (\ref{eq:lower_integ_1}) to find the optimal density $\lambda^{\star}$ that maximizes the lower bound on the sum spectral efficiency. This is obtained as the solution of the optimization
 \begin{align}
\lambda^{\star}=\arg \max_{\lambda} \lambda \log_2\left(1+ \left(2\sinc\left(\frac{2}{\alpha}\right)\right)^{ \frac{\alpha}{2}} \frac{ N_{\rm r}-1 }{ \left(\lambda \pi (1+R_{\rm d}^2) \right)^{\frac{\alpha}{2}}} \right). \label{eq:max_lambda}
\end{align}
In the high SIR regime, i.e., $\log_2(1+x)\simeq \log_2(x)$, the optimal link density is
\begin{align}
\lambda^{\star}=\frac{2\sinc\left(\frac{2}{\alpha}\right)(N_{\rm r}-1)^{\frac{2}{\alpha}}}{\pi (1+R_{\rm d}^2)}.
\end{align}
This simple relationship confirms the intuition that, with MRC, the maximum link density (spatial packing performance) increases sub-linearly with respect to the number of receive antennas.

\section{Spectral Efficiency with Local CSIR } \label{Section4}
In this section, we analyze the sum spectral efficiency of networks using a successive interference cancellation method with local CSIR.

\subsection{Analtyical Characterization}

We first present an analytical expression of the sum spectral efficiency with local CSIR in the following theorem.

\begin{theorem} \label{Th3} The achievable sum spectral efficiency with local CSIR on the $L$ dominant interferers is 
\begin{align}
&C_{\Sigma}^{{\rm sic}}= \frac{\lambda}{\ln(2)}\int_{1}^{R_{\rm d}}\!\int_{0}^{\infty}\!\frac{\! \left[\!1\!-  \frac{1}{(1+zx^{-\alpha})^{N_{\rm r}}} \!\right]\!\!\mathcal{L}_{\tilde{I}_k}(L;z)}{z \exp\left(\frac{z}{{\rm SNR}}\right)  } \d z \frac{x}{ R_{\rm d}^2-1}\d x,
\end{align} 
where 
\begin{align}
&\mathcal{L}_{\tilde{I}_k}(L;z)= \!\int_{0}^{\infty}\!\!\! e^{ - \pi\lambda  \int_{r^2}^{\infty} \!\!\frac{1}{1+z^{-1} u^{\frac{\alpha}{2} } } \d u    }   \frac{2(\lambda\pi r^2)^{L}}{r \Gamma(L)} e^{-\lambda \pi r^2} \d r.  \label{eq:LP_SIC}
\end{align}
\end{theorem}
\begin{proof}
See Appendix \ref{proof:Th3}.
\end{proof}

The main difference with the expression in Theorem \ref{Th1} is the Laplace transform of the aggregated interference power, which reflects the effect of interference cancellation by ZF-SIC. To provide more intuition on the expression in Theorem \ref{Th3}, it is instructive to consider an example.

{\bf Example 7:} When $\alpha=4$, we have a closed form expression for the Laplace transform of $\tilde{I}_k$ in terms of a Bessel function. Conditioning on the fact that the $L$th nearest interferer's distance is equal to $r$, $d_{k,k_{L}}\!\!=\!r$, this Laplace transform is lower bounded as
\begin{align}
\mathcal{\tilde L}_{\tilde{I}_k}(L;z)\!\!&=\mathbb{E}\left[e^{-z \tilde{I}_k} \mid \{  d_{k,k_{L}}=\!r\}\right] \nonumber \\
\!\!&\geq \exp\left(-z \mathbb{E}[\tilde{I}_k \mid  \{  d_{k,k_{L}}=\!r] \right) \nonumber \\
&=\exp\left(-\frac{z \lambda \pi}{r^2}\right), \label{eq:LP_approx}
\end{align}
where the inequality follows from Jensen's inequality and the last equality is due to Campbell's theorem. By unconditioning (\ref{eq:LP_approx}) with respect to $r$,  we obtain
\begin{align}
\mathcal{ L}_{\tilde{I}_k}(L;z)&\geq
\int_{0}^{\infty}\exp\left(-\frac{z \lambda \pi}{r^2}\right)\exp\left(-\lambda \pi r^2\right)\frac{2( \lambda \pi r^2)^L}{r\Gamma(L)}\d r \nonumber \\
&= \frac{2(\lambda \pi)^{L}z^{L/2}}{\Gamma(L)}B_L\left(2\lambda \pi \sqrt{z}\right),\label{eq:appro_LP_bessel}
\end{align} 
where $B_L(x)$ denotes the modified Bessel function of the first kind. By replacing (\ref{eq:LP_SIC}) into (\ref{eq:appro_LP_bessel}), we have 
\begin{align}
&C_{\Sigma}^{{\rm sic}}\geq  \lambda\!\int_{1}^{R_{\rm d}}\int_{0}^{\infty}\!\frac{\! \left[\!1\!- \frac{1}{(1+zx^{-4})^{N_{\rm r}}}\!\right]\frac{2(\lambda \pi)^{L}z^{L/2}}{\Gamma(L)}B_L\left(2\lambda \pi \sqrt{z}\right)}{z \exp\left(\frac{z}{{\rm SNR}}\right)  } \d z \frac{x}{ R_{\rm d}^2-1}\d x. 
\end{align}
Since this expression involves fewer integrals, it is easier to compute. Further, we observe that, given $d_{k,k}=x$, the sum spectral efficiency improves as $L$ increases since $\frac{2(\lambda \pi)^{L}z^{L/2}}{\Gamma(L)}B_L\left(2\lambda \pi \sqrt{z}\right)$ is an increasing function with respect to $L$. This confirms the intuition that interference cancellation improves the sum spectral efficiency.

\begin{figure}
\centering
\includegraphics[width=4.5in]{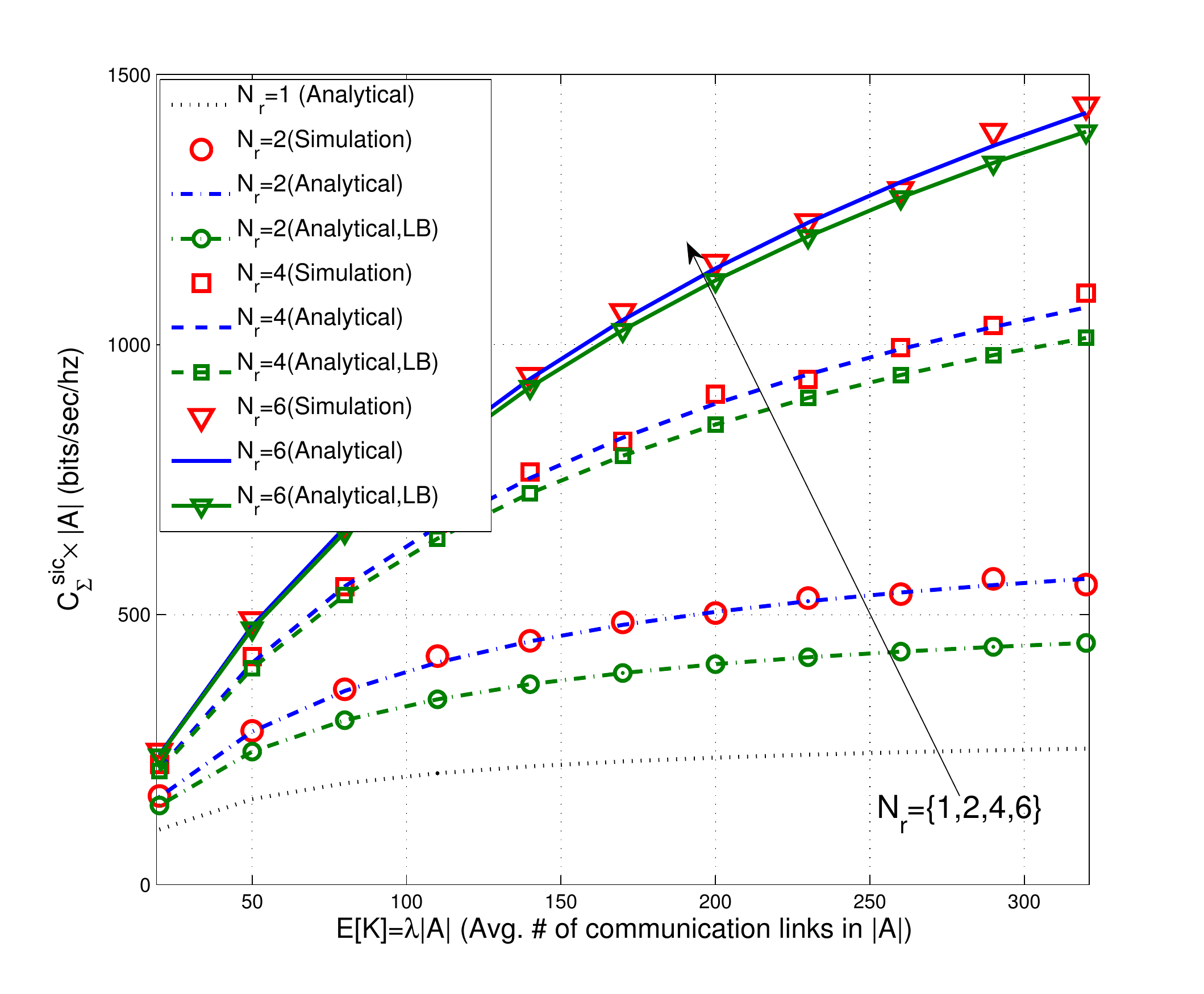} 
\caption{The sum of spectral efficiency achieved with local CSIR when $|\mathcal{A}|=\pi 500^2$ ($m^2$), $\alpha=4$, $R_{\rm d}=50$m, $P=20$ (dBm), and $\sigma^2=-104$ (dBm). }
\label{fig:2}
\end{figure}

Fig. \ref{fig:2} shows a comparison of the sum spectral efficiency achieved by ZF-SIC when $\alpha=4$ and $R_{\rm d}=50$m. The match between analytically and numerically obtained sum spectral efficiencies is excellent. Further, the simple lower bound expression given in (\ref{eq:appro_LP_bessel}) becomes tight as the number of antennas increases. 
   
%%%%%%%%%%%%%%%%%%%%%%%%%%%%%%
%%%%%%%%%%%%%%%%%%%%%%%%%%%%%% 
\subsection{Scaling Law}
By simplifying the exact expression in Theorem \ref{Th3}, we derive the scaling law of the network with local CSIR.

 \begin{theorem}[Scaling law with local CSIR] \label{Th4} Assume that $N_{\rm r}=c\lambda^{\beta}$ for some $c>0$ and $\beta\geq 0$. Then, in interference limited networks, the spectral efficiency of the typical link scales as
\begin{align}
  \frac{C_{\Sigma}^{{\rm sic}}}{\lambda} = \Theta\left(  \log_2\left(1+\lambda^{\frac{\alpha}{2}(\beta-1)}\right)\right),
\end{align}
as $\lambda \rightarrow \infty$.
\end{theorem}

%\left\{
%%\begin{array}{l l}
%  0 & \textrm{when $\beta<1$} \\
% \Theta(1) & \textrm{when $\beta =1$} \\
%  \Theta\left(\frac{\alpha}{2}(\beta-1)\log(\lambda)\right) & \textrm{when $\beta >1$}.
%\end{array} \right.
%\end{align}

\begin{proof}
We begin with the proof of $ \frac{C_{\Sigma}^{{\rm sic}}}{\lambda} = \Omega\left(  \log_2\left(1+\lambda^{\frac{\alpha}{2}(\beta-1)}\right)\right)$ by deriving a lower bound with a closed from. Applying Lemma \ref{lem2}, we obtain the following lower bound on the sum spectral efficiency achieved by ZF-SIC:
\begin{align}
\lambda\mathbb{E}\left[\log_2\!\left(1+\frac{\tilde{H}_{k,k}}{d_{k,k}^{\alpha}{\tilde I}_k  }\right) \right] \geq \lambda\log_2\left(1\!+\!\frac{ e^{\mathbb{E}\left[\ln(\tilde{H}_{k,k})\right]} }{ \frac{2}{(2+\alpha)}R^{\alpha}_{{\rm d} }\mathbb{E}\left[{\tilde I}_k\right] }\right), \label{eq:lower_1}
\end{align}
where we use the fact that $\mathbb{E}\left[d_{k,k}^{\alpha}\right]=\frac{2\left(R_{\rm d}^{\alpha+2}-1\right)}{(\alpha-2)(R_{\rm d}^2 -1)}\leq \frac{2}{2+\alpha}R_{\rm d}^{\alpha}$. Furthermore, from (\ref{eq:mean_sig_mrc}), we know
\begin{align}
e^{\mathbb{E}\left[\ln(\tilde{H}_{k,k})\right]}= e^{\ln(2)+\psi\left(\frac{N_{\rm r}}{2}\right)}. \label{eq:mean_sig1}
\end{align}
Next, we need to compute the expectation of $\tilde{I}_k$. Conditioned on $d_{k,k_L}=r$, the aggregated interference power from the disk with radius $r$ is 
\begin{align}
\mathbb{E} \left[\tilde{I}_k \mid \{d_{k,k_{L}}=r\}\right] =  \frac{2\pi \lambda}{\alpha-2}  r^{2-\alpha}. \label{eq:cond_power_avg1}
\end{align}
Unconditioning with respect to the distance distribution of $d_{k,k_{L}}$ given in \cite{Haenggi}, we obtain the averaged interference power as
\begin{align}
\mathbb{E}[\tilde{I}_k]&=\mathbb{E}_r\mathbb{E} \left[\tilde{I}_k \mid \{d_{k,k_{L}}=r\}\right]%&= \frac{2\pi \lambda \Gamma(1+\frac{\alpha}{2})}{\alpha-2}\b E\left[r^{2-\alpha}\right]
\nonumber \\
&=\frac{2\pi \lambda}{\alpha-2} \int_{0}^{\infty} r^{2-\alpha}\exp\left(-2\pi \lambda r^2\right)\frac{2(2\pi \lambda r^2)^{L}}{r\Gamma(L)}\d r  \nonumber \\
&=\frac{(2\pi \lambda)^{\frac{\alpha}{2}}\Gamma(1-\frac{\alpha}{2}+L)}{(\alpha-2)\Gamma(L)}. \label{eq:mean_inter}
\end{align}
Plugging (\ref{eq:mean_sig1}) and (\ref{eq:mean_inter}) into (\ref{eq:lower_1}), we have the lower bound 
\begin{align}
C_{\Sigma}^{\rm sic} \geq  \lambda\log_2\!\left(1\!+\!\frac{ e^{\left[\ln(2)+\psi(N_{\rm r})  \right] 
}  }{  \frac{2}{2+\alpha}R^{\alpha}_{{\rm d} }\frac{(2\pi \lambda)^{\!\frac{\!\alpha}{2}}\Gamma(1-\frac{\alpha}{2}+L)}{(\alpha-2)\Gamma(L)} }\right).
\end{align}
Since $e^{\ln(2)+\psi(N_{\rm r})}\geq  N_{\rm r}-1$, we can rewrite the lower bound as follows:
\begin{align}
&C_{\Sigma}^{\rm sic} \geq  \lambda \log_2\!\left(1\!+\frac{1}{\frac{2^{\!\frac{\!\alpha}{2}+1}}{(\alpha-2)(\alpha+2)}}  \frac{ (N_{\rm r}-1) \Gamma(L)}{(\lambda \pi R_{\rm d}^2)^{\!\frac{\!\alpha}{2}}\Gamma(1-\frac{\alpha}{2}+L) }\right). \label{eq:sic_scaling}
\end{align}
To reach the scaling law, we use the following inequality in \cite{Kershaw}
\begin{align}
\frac{ \Gamma(L)}{ \Gamma\left(1-\frac{\alpha}{2}+L\right) } \geq \left(L-\frac{\alpha}{2}\right)^{\frac{\alpha}{2}-1},
\end{align}
for $L\geq 1$ and $\alpha\geq 2$. Since the receiver is able to cancel the $L=N_{\rm r}-1$ nearest interferers, using the inequality, the lower bound is given by 
\begin{align}
C_{\Sigma}^{\rm sic}  &\geq  \lambda\log_2\!\left(1\!+\frac{1 }{\frac{2^{\!\frac{\!\alpha}{2}+1}}{(\alpha-2)(\alpha+2)}}  \frac{ (N_{\rm r}-1)\left(N_{\rm r}-1-\frac{\alpha}{2}\right)^{\frac{\alpha}{2}-1}  }{(\lambda \pi R_{\rm d})^{\!\frac{\!\alpha}{2}}   } \right) \nonumber \\
%&\geq \lambda\log_2\!\left(\frac{1 }{\frac{2^{\!\frac{\!\alpha}{2}+1}}{(\alpha-2)(\alpha+2)}}  \frac{ (N_{\rm r}-1)\left(N_{\rm r}-1-\frac{\alpha}{2}\right)^{\frac{\alpha}{2}-1}  }{(\lambda \pi R_{\rm d})^{\!\frac{\!\alpha}{2}}   } \right) \nonumber \\
&= \lambda\log_2\!\left(1+\frac{1 }{\frac{2^{\!\frac{\!\alpha}{2}+1}}{(\alpha-2)(\alpha+2)}}  \frac{ ( \lambda^{\beta}\pi R_{\rm d}^2-1)\left(\lambda^{\beta}\pi R_{\rm d}^2-1-\frac{\alpha}{2}\right)^{\frac{\alpha}{2}-1}  }{(\lambda \pi R_{\rm d})^{\!\frac{\!\alpha}{2}}   } \right),
\end{align}
where the last equality comes from the assumption that the number of antennas can be scaled with the density up to linearly $N_{\rm r} = c\lambda^{\beta}$. As the density goes to infinity, 
\begin{align}
\lim_{\lambda \rightarrow \infty} \frac{C_{\Sigma}^{\rm sic} }{\lambda}  &=\Omega(\log_2(1+\lambda^{\frac{\alpha}{2}(\beta-1)}))  
%\lim_{\lambda \rightarrow \infty} \log_2\!\left(\frac{1 }{\frac{2^{\!\frac{\!\alpha}{2}+1}}{(\alpha-2)(\alpha+2)}}  \frac{ (c_1\lambda\pi R_{\rm d}^2-1)\left(c_1\lambda\pi R_{\rm d}^2-1-\frac{\alpha}{2}\right)^{\frac{\alpha}{2}-1}  }{(\lambda \pi R_{\rm d})^{\!\frac{\!\alpha}{2}}   } \right) \nonumber \\
%&=  \log_2\left(\frac{1 }{\frac{2^{\!\frac{\!\alpha}{2}+1}}{(\alpha-2)(\alpha+2)}} c_1^{\frac{\alpha}{2}}\right)\nonumber \\
%&=  \log_2\left(   \frac{1}{\frac{2^{\!\frac{\!\alpha}{2}+1}}{(\alpha-2)(\alpha+2)}}  \right)+\frac{\alpha}{2}\log_2\left( c_1\right).
.\label{eq:sic_scaling_final}
\end{align}

%As a result, the sum spectral efficiency at least linearly increases with both the density and the path loss exponent, provided that the number of antennas linearly scales with the average number of interferers inside the communication area, $N_{\rm r}= c_1\lambda\pi R_{\rm d}^2 $. 

Now, let us prove $ \frac{C_{\Sigma}^{{\rm sic}}}{\lambda} = O\left(  \log_2\left(1+\lambda^{\frac{\alpha}{2}(\beta-1)}\right)\right)$ by deriving an upper bound with a closed from. Applying Lemma \ref{lem2}, in the interference limited regime, an upper bound on the sum spectral efficiency is given by
\begin{align}
C_{\Sigma}^{\rm sic} \leq \lambda\mathbb{E}_{d_{k,k}} \log_2\left(1+\frac{2\left(1-R_{\rm d}^{2-\alpha}\right)}{(\alpha-2)(R_{\rm d}^2-1)}N_{\rm r}\mathbb{E}\left[\frac{1}{\tilde{I}_k}\right]\right),\label{eq:upper_sic_1}
\end{align}
where we used the facts that $\tilde{H}_{k,k}$, $d_{k,k}$, and ${\tilde I}_k$ are independent and $\mathbb{E}[{\tilde H}_{k,k}] =N_{\rm r}$ and $\mathbb{E}\left[d_{k,k}^{-\alpha}\right]=\frac{2\left(1-R_{\rm d}^{2-\alpha}\right)}{(\alpha-2)(R_{\rm d}^2-1)}$. Thus, we need to compute a negative moment of ${\tilde I}_k$. Using the fact that $\mathbb{E}\left[\frac{1}{X}\right]=\int_{0}^{\infty}\mathbb{E}[e^{-sX}] \d s$, the negative moment is
\begin{align}
\mathbb{E}\left[\frac{1}{{\tilde I}_k}\right]&=\int_{0}^{\infty}\mathbb{E}\left[e^{-s{\tilde I}_k}\right] \d s \nonumber\\
&\stackrel{(a)}{=}\int_{0}^{\infty}\mathbb{E}_r\left[e^{-\lambda \pi \int_{r^2}^{\infty}\frac{1}{1+s^{-1}u^{\frac{\alpha}{2}}} \d u} \right]\d s \nonumber \\
&=\mathbb{E}_r\left[ \int_{0}^{\infty} e^{-\lambda \pi \frac{2}{\alpha-2}r^2 sr^{-\alpha}{}_2F_1\left(1,1-\frac{2}{\alpha},2-\frac{2}{\alpha};sr^{-\alpha}\right)   } \d s \right]
\nonumber \\
&\stackrel{(b)}{=}\mathbb{E}_r\left[ \int_{0}^{\infty} r^{\alpha} e^{-\lambda \pi \frac{2}{\alpha-2} v r^2~{}_2F_1\left(1,1-\frac{2}{\alpha},2-\frac{2}{\alpha};-v\right)   } \d v \right]
\nonumber \\
&= \int_{0}^{\infty} \mathbb{E}_r\left[ r^{\alpha}e^{-\lambda \pi \frac{2}{\alpha-2} v r^2~{}_2F_1\left(1,1-\frac{2}{\alpha},2-\frac{2}{\alpha};-v\right)   } \right]\d v \nonumber \\
&\stackrel{(c)}{=}  \frac{1}{(2\pi\lambda)^{\frac{\alpha}{2}} }\frac{\Gamma\left(N_{\rm r}-1+\frac{\alpha}{2}\right)}{\Gamma(N_{\rm r}-1)}\int_{0}^{\infty} \!\!\!\!\frac{1}{\left[1+\frac{2v~{}_2F_1\left(1,1-\frac{2}{\alpha},2-\frac{2}{\alpha};-v\right)}{\alpha-2}\right]^{\frac{\alpha}{2}+N_{\rm r}-1}} \d v,
\end{align}
where (a) follows from the probability generating functional of the PPP, (b) is due to the variable change $v=sr^{-\alpha}$ and ${}_2F_1(\cdot)$ denotes the Gauss hypergeometric function, (c) follows from the distance distribution of the $(N_{\rm r}-1)$th nearest interferer from the $k$th receiver, given in \cite{Haenggi}. Using the following inequalities: 
\begin{align}
\int_{0}^{\infty}\frac{1}{\left[1+\frac{2v~{}_2F_1\left(1,1-\frac{2}{\alpha},2-\frac{2}{\alpha};-v\right)}{\alpha-2}\right]^{\frac{\alpha}{2}+N_{\rm r}-1}} \d v &=\int_{0}^{\infty} \!\!\!\!\frac{1}{\left[1+v^{\frac{2}{\alpha}}\int_{v^{-\frac{2}{\alpha}}}^{\infty}\frac{1}{1+u^{\frac{\alpha}{2}}} \d u\right]^{\frac{\alpha}{2}+N_{\rm r}-1}} \d v
\nonumber \\ 
&\leq \int_{0}^{\infty} \!\!\!\!\frac{1}{\left[1+v^{\frac{2}{\alpha}} \right]^{\frac{\alpha}{2}+N_{\rm r}-1}} \d v  \nonumber \\ 
&\leq \int_{1}^{\infty} \!\!\!\!\frac{1}{(v^{\frac{2}{\alpha}})^{\frac{\alpha}{2}+N_{\rm r}-1}} \d v,
\end{align}
for $\alpha>2$ and $N_{\rm r}>1$, we get the upper bound on the negative moment as follows:
\begin{align}
\mathbb{E}\left[\frac{1}{{\tilde I}_k}\right] &\leq  \frac{1}{(2\pi\lambda)^{\frac{\alpha}{2}} }\frac{\Gamma\left(N_{\rm r}-1+\frac{\alpha}{2}\right)}{\Gamma(N_{\rm r}-1)}\int_{1}^{\infty} \!\!\!\!\frac{1}{(v^{\frac{2}{\alpha}})^{\frac{\alpha}{2}+N_{\rm r}-1}} \d v \nonumber\\
&=  \frac{1}{(2\pi\lambda)^{\frac{\alpha}{2}} }\frac{\Gamma\left(N_{\rm r}-1+\frac{\alpha}{2}\right)}{\Gamma(N_{\rm r}-1)}\frac{\alpha}{2(N_{\rm r}-1)}.\label{eq:inter_upper_sic}
%&=\mathbb{E}_r\left[ \int_{0}^{\infty}e^{-\lambda \pi s  \frac{2}{\alpha-2}  r^{2(1-\frac{\alpha}{2}) }  }  \d s\right] \nonumber\\
%&= \frac{1}{\lambda \pi\frac{2}{\alpha-2}}\mathbb{E}_r\left[r^{2(\frac{\alpha}{2}-1) }\right] \nonumber\\
%&= \frac{(\alpha-2)}{2\pi \lambda} (2\pi\lambda)^{1-\frac{\alpha}{2}}\frac{\Gamma\left(N_{\rm r}-2+\frac{\alpha}{2}\right)}{\Gamma(N_{\rm r}-1)},(c) comes from the fact that $e^{-x^{1-\frac{\alpha}{2}}}$ is a concave function with respect to $x$ for $\alpha>2$, and 
\end{align}
Plugging (\ref{eq:inter_upper_sic}) into (\ref{eq:upper_sic_1}), from the fact that $\frac{\Gamma\left(N_{\rm r}-1+\frac{\alpha}{2}\right)}{\Gamma(N_{\rm r}-1)}\leq (N_{\rm r}-1)^{\frac{\alpha}{2}}$ we get the upper bound
\begin{align}
C_{\Sigma}^{\rm sic} &\leq \lambda \log_2\left(1+\frac{\alpha}{ 2^{\frac{\alpha}{2}+1}} \frac{(N_{\rm r}-1)^{\frac{\alpha}{2}}}{(\pi\lambda )^{\frac{\alpha}{2}} } \frac{N_{\rm r}}{N_{\rm r}-1}\frac{2\left(1-R_{\rm d}^{2-\alpha}\right)}{(\alpha-2)(R_{\rm d}^2-1)}\right).\label{eq:upper_sic_2}
%&= \lambda\log_2\left(1+ \frac{N_{\rm r}(N_{\rm r}-1)^{\frac{\alpha}{2}-1} }{(\lambda \pi R_{\rm d}^2)^{\frac{\alpha}{2}}} \right).  
% \nonumber \\
%&= \lambda\log_2\left(1+\frac{2}{\alpha-2}\frac{ (\lambda \pi R_{\rm d}^2)^{\beta}\left((\lambda \pi R_{\rm d}^2)^{\beta}+\frac{\alpha}{2}\right)^{\frac{\alpha}{2}} }{(\lambda \pi R_{\rm d}^2)^{\frac{\alpha}{2}}} \right), \label{eq:upper_sic_2}
 % \nonumber \\
%&\leq \lambda\log_2\left(\frac{2}{\alpha-2}\right)+\lambda\frac{\alpha}{2}\log_2\left( c_2^{\frac{\alpha}{2}+1}\right), 
\end{align}
Assuming that $N_{\rm r} =c\lambda^{\beta}  $, as the density goes to infinity, we get
\begin{align}
  \frac{C_{\Sigma}^{\rm sic} }{\lambda} =  O\left(\log_2\left(1+\lambda^{\frac{\alpha}{2}(\beta-1)}\right)\right),
\end{align}
as $\lambda \rightarrow \infty$. This completes the proof.

%Instead of directly computing the negative moment of ${\tilde I}_k$, we calculate the negative moment of the interference power without fading from the $N_{\rm r}$th nearest interferer, ignoring the rest of interferers. It follows from $\mathbb{P}\left[\frac{1}{\tilde{ I}_k}\geq u\right] \leq \mathbb{P}\left[\frac{1}{d_{k,N_{\rm r}}^{-\alpha}}\geq u\right]$, for $u\geq 0$, that
%\begin{align}
%\mathbb{E}\left[\frac{1}{\tilde{ I}_k}\right]&=\int_{0}^{\infty}\mathbb{P}\left[\frac{1}{\tilde{ I}_k}\geq u\right] \d u
%\nonumber\\
%&\leq\int_{0}^{\infty}\mathbb{P}\left[d_{k,N_{\rm r}}^{\alpha}\geq u\right] \d u
%\nonumber\\
%&=\mathbb{E}\left[d_{k,N_{\rm r}}^{\alpha}\right]
%\nonumber\\
%&=\!\int_{0}^{\infty} r^{\alpha} \frac{2(\lambda\pi r^2)^{N_{\rm r}}}{r \Gamma(N_{\rm r})} e^{-\lambda \pi r^2} \d r \nonumber \\
%&=\frac{\frac{\Gamma\left(N_{\rm r}+\frac{\alpha}{2}\right)}{\Gamma\left(N_{\rm r}\right)}}{(\lambda \pi)^{\frac{\alpha}{2}}}, \label{eq:inter_upper_sic}
%\end{align}

\end{proof}

%This method provides an upper bound on $\mathbb{E}\left[\frac{1}{\tilde{ I}_k}\right]$ because $\frac{1}{d_{k,N_{\rm r}}^{-\alpha}}$ is statistically dominant of $\frac{1}{\tilde{ I}_k}$, i.e., $\mathbb{P}\left[\frac{1}{\tilde{ I}_k}\geq u\right] \leq \mathbb{P}\left[\frac{1}{d_{k,N_{\rm r}}^{-\alpha}}\geq u\right]$ for $u\geq 0$. 

This scaling law is remarkable in that the sum spectral efficiency improves linearly with the density even if the number of receive antennas scales up linearly with the density. This indicates that the linear capacity scaling law is achievable with less receive antennas than MRC.
Furthermore, when the number of antenna increases like $\lambda^{\beta}$ for $\beta >1$, the sum spectral efficiency increases super-linearly with the density with a multiplicative gain of $\frac{\alpha}{2}(\beta-1)$, which is proportional to the path loss exponent $\alpha$. This multiplicative gain in the capacity scaling comes from the performance improvements by the dominant interference cancellation. 
These observations advocate that, without transmit cooperation, near-capacity-achieving point-to-point coding is able to provide significant performance gain by an appropriate combination of strong interference cancellation and treating weak interference
as noise. A similar observation was also made in single antenna ad hoc systems using simultaneous decoding of strong interfering signals at receivers \cite{Baccelli_Tse}.

It is also worth to note that the scaling law attained by ZF-SIC can be obtained with a constant rate loss when partial zero-forcing (PZF) in \cite{Jindal_PZF} is applied. For example, when $\alpha=4$, we choose the number of interferers being cancelled to be $\frac{N_{\rm r}}{2}$ while boosting the desired signal power using the remaining antenna degrees of freedom $\frac{N_{\rm r}}{2}$. This case can equivalently be interpreted to the case where receivers apply ZF-SIC with $\frac{N_{\rm r}}{2}$ receive antennas. %PZF is more preferable than ZF-SIC when receivers cannot obtain the modulation and coding methods used by the nearest transmitters.

%\subsection{Maximum Intensity}
%In the network, one interesting question would be what is the maximum traffic intensity $\lambda$ for a set of given system parameters to maximize the sum of  spectral efficiencies. For this purpose, by leveraging the lower bound of the sum spectral efficiency in (\ref{eq:sic_scaling}), we characterize the maximum allowable traffic intensity in the following corollary. 
%
%\begin{corollary}\label{col1}
%\label{intensity}
%When $\alpha=4$, the maximum allowable intensity is
% \begin{align}
%\lambda^{\star}=\sqrt{\frac{ e^{-\gamma  
%}  N_{\rm r}}{  2\pi^{2} \Gamma(N_{\rm r}-2)  }}. \label{eq:max_lambda}
%\end{align}
%\end{corollary}
%\proof By taking the derivative of (\ref{eq:sic_scaling}) with respective to $\lambda$ and solving the first order Karush-Kuhn-Tucker condition, we find the maximum value of $\lambda$, which arrives at (\ref{eq:max_lambda}).
%\endproof
% 
%
%This result allows to find the optimal traffic intensity for a given system parameters $N_{\rm r}$ and $R_{\rm d}$. For the given link distance $d_{k,k}=d$ and the large enough $N_{\rm r}$, it is further approximately rewritten as
% \begin{align}
%\lambda^{\star}\simeq \frac{ N_{\rm r} }{\pi d^2}.\label{eq:max_lambda_simple}
%\end{align}
%This simple relationship confirms the intuition that the larger link distance diminishes the maximum allowable density. Meanwhile, as observed in Corollary \ref{col1}, with multiple receive antennas, the maximum  intensity (spatial packing) increases linearly.

\section{Effects of Antenna Correlation and Bounded Pathloss Function} \label{Section5}
In this section, we analyze the impact of receive antenna correlation and bounded path loss function on the sum of spectral efficiencies and its scaling behavior when the receiver is aware of direct CSIR.

\subsection{Antenna Correlation Effect}
\subsubsection{Correlation Model}
To incorporate the effect of correlation, we model the channel vector ${\bf h}_{k,{\ell}}$ as
\begin{align}
{\bf h}_{k,{\ell}}={\bf C}^{\frac{1}{2}}{\bf \tilde h}_{k,{\ell}},
\end{align}
where ${\bf C} \in \mathbb{R}^{N_{\rm r}\times N_{\rm r}}$ denotes a receive antenna correlation matrix, which is assumed to have the positive ordered eigenvalues $\{\mu_1,\ldots,\mu_{r}\}$, $\mu_n\geq \mu_m$ for $n<m$  i.e., ${\rm rank}({\bf C})=r$ with $r\leq N_{\rm r}$. Further, the entries of ${\bf \tilde h}_{k,{\ell}}$ are IID complex Gaussian random variables, each with zero mean and unit variance. The eigenvalues can be different depending on the antenna structure. For example, it has been shown experimentally that the spatial correlation matrix of a uniform linear array antenna is well represented by the exponential model introduced in \cite{Loyka}. For mathematical convenience, we assume that ${\bf C}$ is fixed and compute the ergodic rate with respect to the fadings. This assumption is valid because the second-order statistics of ${\bf C}$ change slowly relative to the fadings in time.

The following lemma quantifies us how the antenna correlation changes the effective fading distributions in both the direct and the interfering links.

\begin{lemma} \label{lem:correlation}[Fading distributions with antenna correlation] For the antenna correlation matrix ${\bf C}$, the fading distribution of $H_{k,k}=| {\bf  w}_k^{H}{\bf  h}_{k,k} |^2$ is the sum of independent exponential random variables with means $\left\{ \mu_1,\ldots, \mu_{r} \right\}$. Further, the fading distribution for the interfering link, $H_{k,\ell}=|{\bf  w}_k^{H}{\bf h}_{k,\ell}|^2$ for $k \neq \ell$ satisfies:\begin{align}
\mathbb{P}\left[H_{k,\ell}>x\right] \leq \exp\left(-\frac{x}{\mu_{1}}\right).
\end{align}
\end{lemma}
 \begin{proof}
See Appendix \ref{proof:lem_corr}
\end{proof}

\subsubsection{A Lower Bound}
Leveraging Lemma \ref{lem:correlation}, we now derive a lower bound on the sum spectral efficiency. The corresponding upper bound is obtained when the antennas are uncorrelated, which is given in Theorem \ref{Th1}.

\begin{theorem} \label{Th5}
Assume that the correlation matrix ${\bf C} \in \mathbb{R}^{N_{\rm r}\times N_{\rm r}}$ has non-zero eigenvalues $\{\mu_1,\ldots,\mu_{{\rm rk}({\bf C})}\}$ where ${\rm rk}({\bf C})\leq N_{\rm r}$. The sum spectral efficiency with direct CSIR is lower bounded by
\begin{align}
C_{\Sigma}^{{\rm mrc}}  \geq \lambda \int_{1}^{R_{\rm d}}\int_{0}^{\infty}\frac{\left( 1- \prod_{n=1}^{{\rm rk}({\bf C})} \frac{1}{1+\mu_n zx^{-\alpha}}\right)e^{-\frac{\mu_1^{\frac{2}{\alpha}}\lambda \pi }{\sinc\left(\frac{2}{\alpha}\right)}z^{\frac{2}{\alpha}} }}{z \exp\left(\frac{z}{{\rm SNR}}\right)  } \d z \frac{2x}{R_{{\rm d}}^2-1}\d x.
 \end{align}
\end{theorem} 

\proof The proof is direct from the proof of Theorem \ref{Th1} and \ref{Th2}, replacing the Laplace transforms of $H_{k,k}$ and $H_{k,\ell}$ considering the antenna correlation, which are distributed per Lemma \ref{lem:correlation}.\endproof
 
This shows that the eigenvalues of the antenna correlation matrix affect the ergodic spectral efficiency by changing both the desired signal power and the aggregated interference power. As a special case, by setting $\mu_n=1$ for $n\in\{1,\ldots,N_{\rm r}\}$, we then recover the exact expression of the sum spectral efficiency given in Theorem \ref{Th1}. It is interesting to observe that the performance degrades as the condition number of the correlation matrix, $\kappa({\bf C})=\frac{\mu_1}{\mu_{r}}$, becomes larger. This implies that the sum spectral efficiency decreases in highly correlated antenna structures.

\begin{figure}
\centering
\includegraphics[width=4.5in]{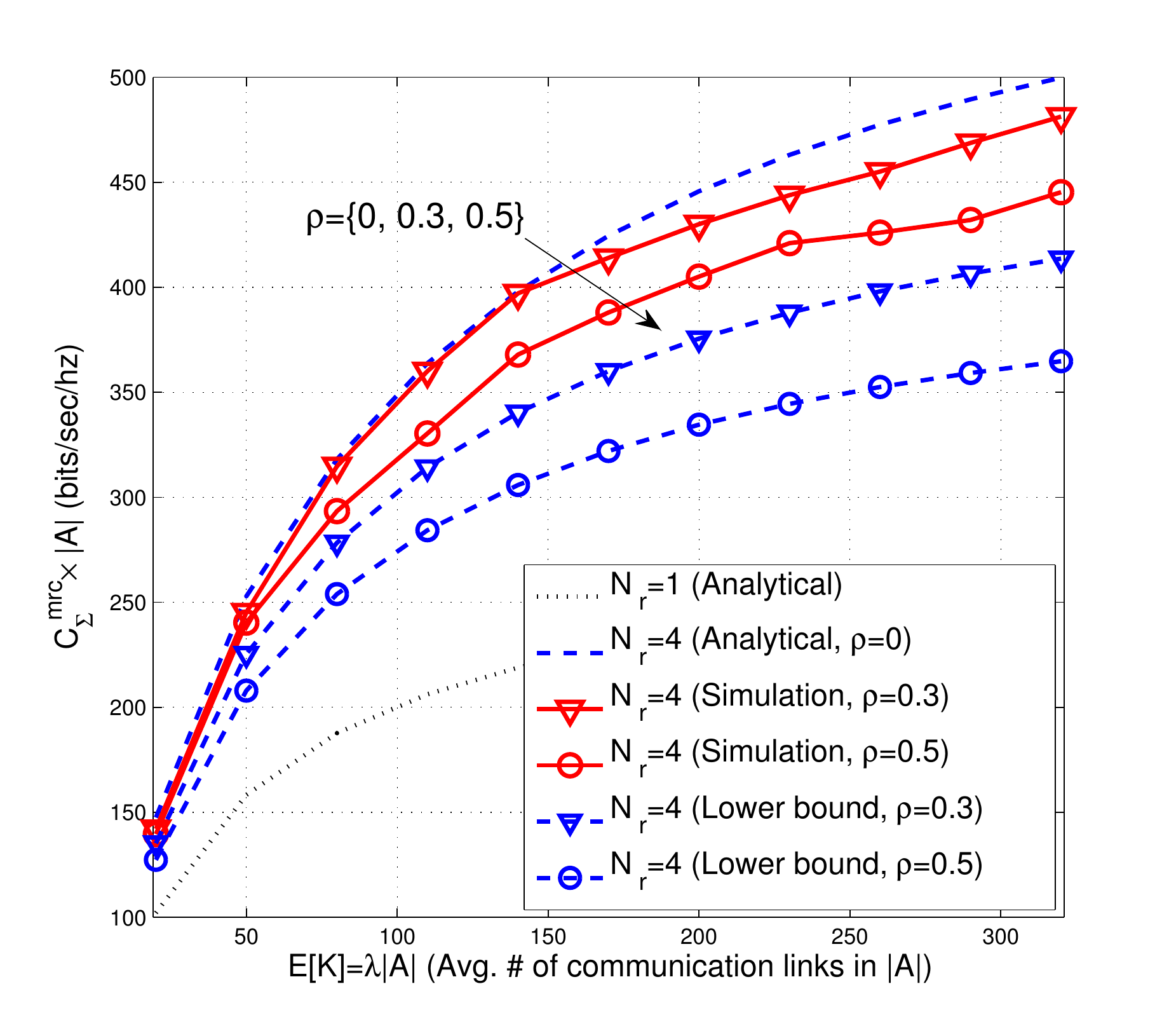} 
\caption{The sum of spectral efficiencies achieved with direct CSIR when atenna correlation is present for $|\mathcal{A}|=\pi 500^2$, $\alpha=4$, $R_{\rm d}=50 m$, and ${\rm SNR}=4.99513\times 10^9$.}
\label{fig:3}
\end{figure}

Fig. \ref{fig:3} illustrates the sum spectral efficiency when $N_{\rm r}=4$ according to different antenna correlation parameters. Using the exponential antenna correlation model in \cite{Loyka}, we define a correlation matrix ${\bf C}$ as
\begin{align}
{\bf C}=\left[%
\begin{array}{cccc}
 1 &\rho &\rho^2&\rho^3  \\
 \rho &1&\rho& \rho^2  \\
 \rho^2 &\rho& 1&\rho  \\
 \rho^3 &\rho^2& \rho & 1  \\
\end{array}%
\right],
\end{align}
where $\rho$ denotes a correlation parameter between two adjacent receive antennas. As shown in Fig. \ref{fig:3}, the sum spectral efficiency decreases as the antenna  correlation value increases. It is notable that in the lower density regime, the sum spectral efficiency degradation due to antenna correlation is negligible. Whereas, in the denser density regime, the antenna correlation deteriorates the performance.  

\subsubsection{Scaling Law} 
We derive a lower bound on scaling law when direct CSIR is known, considering the antenna correlation effects.

 \begin{corollary}[Scaling law with antenna correlation] \label{cor1} Assume that  ${\rm rk}({\bf C})= c\lambda^{\beta}$ for some $c>0$.
The ergodic spectral efficiency of a typical link asymptotically scales as follows:
\begin{align}
 \frac{C_{\Sigma}^{{\rm mrc}} }{\lambda}=\Theta\left(\log_2\left(1+\lambda^{\beta-\frac{\alpha}{2}}\right)\right),
\end{align} 
as $\lambda \rightarrow \infty$.
\end{corollary}
 \begin{proof}
The proof uses arguments similar to those of the proof of Theorem \ref{lem2}. Since the interference power is changed by a constant factor $\mu_1$, it does not affect the scaling law. The difference is in the computation of $\mathbb{E}\left[ \ln\left({H}_{k,k}\right)\right] $ when taking the antenna correlation effect into account. Since $H_{k,k}$ is a weighted sum of exponential random variables, we can define $H_{k,k}=\sum_{n=1}^{{\rm rk}\left({\bf C}\right)}\mu_n {X}_n$ where $X_n$ is an IID exponential random variable with mean one. With this, we compute a lower bound on $\mathbb{E}\left[ \ln\left({H}_{k,k}\right)\right] $ as
\begin{align}
\mathbb{E}\left[ \ln\left({H}_{k,k}\right)\right] &=\mathbb{E}\left[ \ln\left(\sum_{n=1}^{{\rm rk}\left({\bf C}\right)}\mu_n {X}_n\right)\right] \nonumber\\
&\stackrel{(a)}{\geq} \mathbb{E}\left[ \ln\left\{{\rm rk}\left({\bf C}\right)\left(\prod_{n=1}^{{\rm rk}\left({\bf C}\right)} \mu_n{X}_n\right)^{\frac{1}{r}}\right\}\right] \nonumber\\
&=\ln({\rm rk}\left({\bf C}\right))+\frac{1}{{\rm rk}\left({\bf C}\right)}\sum_{n=1}^{{\rm rk}\left({\bf C}\right)} \left\{\ln(\mu_n)+\mathbb{E}\left[\ln(X_n)\right]\right\}\nonumber\\
&\stackrel{(b)}{=}\ln({\rm rk}\left({\bf C}\right))+\frac{1}{{\rm rk}\left({\bf C}\right)}\sum_{n=1}^{{\rm rk}\left({\bf C}\right)} \ln(\mu_n) -\gamma, \label{eq:low_desired_correlation}
\end{align}
where (a) follows from the arithmetic geometric mean inequality and (b) follows from the fact that $X_n$ is Chi-squared distributed with degrees of freedom two and the definition of the Euler-Mascheroni's constant $\gamma$. With this, for the given $I_k$, a lower bound on the sum spectral efficiency is given by
\begin{align}
\mathbb{E}\left[\log_2\left(1+\frac{H_{k,k}}{I_k}\right)\right]&\geq \log_2\left(1+\frac{e^{\mathbb{E}[\ln(H_{k,k})]}}{I_k}\right) \nonumber \\
&= \lambda \log_2\left(1+e^{-\gamma} e^{\frac{\sum_{n=1}^{{\rm rk}\left({\bf C}\right)} \ln(\mu_n)}{{\rm rk}\left({\bf C}\right)}} \frac{ {\rm rk}\left({\bf C}\right)}{I_k}\right). \label{eq:lower_antecor}
\end{align}
Using the same argument as in the proof of Theorem \ref{Th2}, it is possible to show that the lower bound in (\ref{eq:lower_antecor}) scales as $\lim_{\lambda \rightarrow \infty}\frac{C_{\Sigma}^{{\rm mrc}} }{\lambda}=\Omega\left(\log_2\left(1+\lambda^{\beta-\frac{\alpha}{2}}\right)\right)$. Since the upper bound corresponds to the case of no antenna correlation shown in Theorem \ref{Th2}, we complete the proof of Corollary \ref{cor1}.%
%
% As a result, the MRC method achieves linear scaling law of the ergodic spectral efficiency when the receive antennas are correlated if the rank of the correlation matrix $r={\rm rank}({\bf C})$ increases linearly in the density of links.
\end{proof}

Corollary \ref{cor1} demonstrates that antenna correlation does not affect the scaling law of the network if the correlation matrix has a full rank. Nevertheless, the antenna correlation decreases the SIR by the factor of $e^{\frac{\sum_{n=1}^{{\rm rk}\left({\bf C}\right)} \ln(\mu_n)}{{\rm rk}\left({\bf C}\right)}-\gamma}$ compared to the uncorrelated case as observed in Fig. \ref{fig:3}.

\subsection{Effect of the Path Loss Function}
Up to this point, we have characterized the scaling law by using the path loss function $\|x\|^{-\alpha}$. In this section, we analyze the impact on the achievable scaling law of using the bounded path loss function ($\min\left\{1,\|x\|^{-\alpha} \right\}$). Unlike the path loss function used in the previous sections, this bounded path loss function ensures that the mean power of the aggregated interference is finite. Using this bounded path loss function, we provide the scaling law in networks with direct CSIR in the following corollary. 

 \begin{corollary} \label{cor2} In the interference limited regime, if $N_{\rm r}=c\lambda^{\beta}$ with $c>0$, then 
\begin{align}
\frac{C_{\Sigma}^{{\rm mrc}}}{\lambda} &=\Omega\left(\log_2\left(1+\lambda^{\beta-\frac{\alpha}{2}}\right)\right),
\end{align} 
as $\lambda \rightarrow \infty$.
\end{corollary}

\proof
%The proof of the achievability for $\Omega\left(\log_2\left(1+\lambda^{\beta-\frac{\alpha}{2}}\right)\right)$ is direct from Theorem \ref{Th1} because the singular path loss function causes the higher interference power than that obtained from the bounded path loss function, i.e., $\|x\|^{-\alpha} \geq \min\left\{1,\|x\|^{-\alpha} \right\}$ for all $x$. Therefore, we only need to focus on the proof of $O\left(\log_2\left(1+\lambda^{\beta-\frac{\alpha}{2}}\right)\right)$ by deriving an upper bound on the spectral efficiency when the bounded path loss function is considered. 

From Lemma \ref{lem2}, the sum spectral efficiency is upper bounded as follows:
\begin{align}
&\lambda\mathbb{E}_{H_{k,k},d_{k,k},I_k}\left[\log_2\!\left(1+\frac{H_{k,k}d_{k,k}^{-\alpha}}{I_k}\right) \right]\geq\lambda \log_2\left(1\!+\!\frac{ e^{\mathbb{E}\left[\ln(H_{k,k})\right]-\alpha\mathbb{E}\left[\ln(d_{k,k})\right]}}{ \mathbb{E}\left[I_k\right] }\right).\label{eq:lower_2}
\end{align}
Using the facts that $H_{k,k}$ is a Chi-squared random variable with $2N_{\rm r}$ degrees of freedom, $d_{k,k}$ is uniformly distributed in a disk with radius $R_{\rm d}$, and they are mutually independent, we obtain the following inequality:
\begin{align}
 e^{\mathbb{E}\left[\ln(H_{k,k})\right]-\alpha\mathbb{E}\left[\ln(d_{k,k})\right]}
\geq   (N_{\rm r}-1)e^{\frac{\alpha}{2R_{\rm d}^2}-\frac{\alpha}{2}}.\label{eq:mean_sig_mrc}\end{align}

Next, we compute the mean of the interference power. 
Since we consider a non-singular path loss model, the expectation of the aggregated interference power is bounded. Applying Campbell's theorem \cite{Baccelli_book}, the mean of the aggregated interference power is given by
\begin{align}
\mathbb{E} \left[I_k \right] &=\!\b{E}\left[\!\sum_{d_{k,k_{j}}\in {\Phi}} H_{k,k_j}\min\left\{1,d_{k,k_j}^{-\alpha}\right\} \right] \nonumber \\
&=\lambda 2\pi\mathbb{E}[H_{k,k_j}]\left[\int_{0}^1 r \d r+ \int_{1}^{\infty}r^{-\alpha} r \d r \right] \nonumber\\
&= 2\pi \lambda\left(\frac{1}{2}+\frac{1}{\alpha-2}\right).\label{eq:avg_Intpower_mrc}
\end{align} 
 
Invoking (\ref{eq:mean_sig_mrc}) and (\ref{eq:avg_Intpower_mrc}) into (\ref{eq:lower_2}), we finally obtain the lower bound
\begin{align}
&C_{\Sigma}^{{\rm mrc}}\geq \lambda\log_2\!\left(1\!+\!\frac{(N_{\rm r}-1)e^{\frac{\alpha}{2R_{\rm d}^2}-\frac{\alpha}{2}} }{ 2\pi \lambda\left(\frac{1}{2}+\frac{1}{\alpha-2}\right)   }\right).
\end{align}
Since we assumed $N_{\rm r}=\lambda^{\beta}$, as the density goes to infinity, we get
\begin{align}
\lim_{\lambda \rightarrow \infty}\frac{C_{\Sigma}^{{\rm mrc}}}{\lambda} &\geq  \log_2\!\left(1\!+\!\frac{ e^{\frac{\alpha}{2R_{\rm d}^2}-\frac{\alpha}{2}} }{ 2\pi\left(\frac{1}{2}+\frac{1}{\alpha-2}\right) } \lambda^{\beta-1} \right),
\end{align}
which completes the proof.

\endproof

This corollary contrasts with the capacity scaling result given in Theorem \ref{Th1}. In Theorem \ref{Th1}, we showed that the number of receive antennas should increase super-linearly with the density to maintain the linear scaling law with the direct CSIR. Corollary \ref{cor2}, however, shows that the linear scaling of the sum spectral efficiency is possible with knowledge of CSIR for the direct links, whenever the number of receive antennas scales linearly with the density of links. These different results show that the network performance strongly depends on the chosen path-loss function. A similar observation was reported in recent work \cite{Xinchen_multi_slop} where multi-slope path loss functions change the coverage probability as a function of the base station density in cellular downlink networks.

\section{Conclusion} \label{Section6}
 We considered a network with multiple receive antennas and explored the benefits of exploiting multiple antennas in terms of the sum spectral efficiency. Under two different CSIR assumptions, we derived exact analytical expressions and scaling laws by deriving closed from upper and lower bounds on the sum spectral efficiency. One major implication from our results is that the sum spectral efficiency improves linearly with the density of links when the number of antennas scales with the density in a particular super-linear way. This super linear growth conclusion holds under the assumption of a power law attenuation. For the bounded attenuation, the super linear growth can be replaced by a linear one. When local CSIR is exploited, the sum spectral efficiency improves linearly with the multiplicative factor given by the path loss when the number of receive antenna scales with the density in a linear way. Further, we verified that for correlated channels, a linear scaling is still achievable with direct CSIR as long as the rank of a spatial correlation matrix scales super-linearly with the density. These results show that using multiple antennas is  useful in controlling interference in a  distributed way; thereby providing significant gains in the network scaling.

An interesting direction for future study would be to explore the effects of having multiple antennas at transmitters. For example, the transmit array can be used for maximum ratio transmission or to apply nulls to nearby interferers \cite{Akoum}. Further, by combining interference-aware scheduling algorithms, it would also be interesting to characterize the achievable rate when multiple antennas are used in by transmitters and receivers. Another interesting direction is to extend the results to cellular networks by changing the direct link distance distribution and interference guard regions appropriately. It would also be interesting to consider millimeter wave operation where channel blockages are important and antenna arrays are used only for beamforming \cite{Bai}.

\appendices

%\begin{lemma} \label{lemma:distance}[PPP distance distribution] Let us consider the $k$th D2D receiver as a typical receiver. Given a PPP in a plane with intensity $\lambda$, the distribution of the distance $d_{k,k_{\ell}}$ between the typical receiver and its the $N$th nearest  interferers is the generalized Gamma distribution
%\begin{align}
%f_{d_{k,k_{\ell}}}(r)= \frac{2(\lambda\pi r^2)^{\ell}}{r \Gamma(\ell)} e^{-\lambda \pi r^2} %\label{lemma:distance} ,
%\end{align}
%where $\Gamma(\ell)=\int_{0}^{\infty}e^{-x}x^{\ell-1}\d x$ is the Gamma function.\end{lemma}
%
%\begin{proof}
%See \cite{Haenggi}.
%\end{proof}

%%%%%%%%%%%%%%%%%%%%%%%%%%%%%
%%%%%%%%%%%%%%%%%%%%%%%%%%%%%

\section{Proof of Theorem \ref{Th1}}\label{proof:Th1}
Conditioning by $d_{k,k}=d$, leveraging Lemma \ref{lem1}, the sum spectral efficiency of the $k$th  link can be written in the following integral form:
\begin{align}
&\mathbb{E}\left[\log_2\left(1+\frac{H_{k,k}}{d^{\alpha}I_k+\frac{d^{\alpha}}{{\rm SNR}}}\right)\mid d_{k,k,}=d\right] \nonumber \\&=\frac{1}{\ln(2)}\int_{0}^{\infty}\frac{e^{-\frac{d^{\alpha}}{{\rm SNR}}z}}{z}\left(1-\mathbb{E}\left[e^{-zH_{k,k}}\right]\right)\mathbb{E}\left[e^{-zd^{\alpha}I_k}\right] \d z, \label{eq:condi_ergodic_rate}
\end{align}
where the expectations are taken over $H_{k,k}$ and $I_k$. The Laplace transform of the aggregated interference power $I_k$ evaluated at $d^{\alpha}z$ is computed as
\begin{align}
\mathbb{E}\left[e^{-zd^{\alpha}I_k}\right] &=\mathbb{E}\!\left[e^{-z d^{\alpha} \sum_{\ell \!=1}^{\infty}\!H_{k,k_{\ell}}d_{k,k_{\ell}}^{-\alpha} }  \right]\nonumber  \\
&\stackrel{(a)}{=} \mathbb{E}\left[ \prod_{d_{k,k_{\ell}}\in \Phi} \!\!e^{-z d^{\alpha} H_{k,k_{\ell}}d_{k,k_{\ell}}^{-\alpha}}  \!\right] \nonumber  \\
&\stackrel{(b)}{=}  \mathbb{E}\left[ \prod_{d_{k,k_{\ell}}\in \Phi }  \frac{1}{1+z d^{\alpha} d_{k,k_{\ell}}^{-\alpha}} \right]\nonumber \\
&\stackrel{(c)}{=} \exp\left(-\frac{\lambda \pi d^{2} }{\sinc\left(\frac{2}{\alpha}\right)} z^{\frac{2}{\alpha}}  \right),  \label{eq:LF_inter}
\end{align}
where (a) follows from the independence of ${d}_{k,k_{\ell}}$ and $H_{k,k_j}$, (b) holds because $H_{k,k_{\ell}}$ is exponentially distributed with unit mean $\forall \ell$, (c) follows from the probability generating functional of the PPP and the definition of the sinc function. Plugging (\ref{eq:LF_inter}) into (\ref{eq:condi_ergodic_rate}), we obtain the conditional spectral efficiency of the $k$th link as
\begin{align}
&\mathbb{E}\left[\log_2\left(1+\frac{H_{k,k}}{d^{\alpha}I_k+\frac{d^{\alpha}}{{\rm SNR}}}\right)\mid \{d_{k,k,}=d\}\right] \nonumber \\&=\frac{1}{\ln(2)}\int_{0}^{\infty}\frac{e^{-\frac{d^{\alpha}z}{{\rm SNR}}} }{z}\left(1\!-\!\mathbb{E}\!\left[e^{-z H_{k,k}}\!\right]\!\right)e^{-\frac{\lambda \pi d^{2} }{\sinc\left(\frac{2}{\alpha}\right)} z^{\frac{2}{\alpha}} }  \d z \nonumber \\
&\stackrel{(a)}{=}\frac{\alpha}{2\ln(2)}\int_{0}^{\infty}\frac{e^{-\frac{d^{\alpha}\left(\frac{\sinc\left(\frac{2}{\alpha}\right)u }{\lambda \pi d^{2}}\right)^{\!\!\! \frac{\alpha}{2}}  }{{\rm SNR}}}}{u}\!\left(\!\!1\!-\!\mathbb{E}\!\left[\!e^{-\left(\frac{\sinc\left(\frac{2}{\alpha}\right)u }{\lambda \pi d^{2}}\right)^{\!\!\! \frac{\alpha}{2}} H_{k,k}}\!\right]\!\right)e^{-u }  \d u
\nonumber \\
&\stackrel{(b)}{=}\frac{\alpha}{2\ln(2)}\int_{0}^{\infty}\frac{e^{-\frac{\left(\frac{\sinc\left(\frac{2}{\alpha}\right)u }{\lambda \pi  }\right)^{\!\!\! \frac{\alpha}{2}}  }{{\rm SNR}}}}{u}\!\left(\!1\!- \frac{1}{\left(1+\left(\frac{\sinc\left(\frac{2}{\alpha}\right)u }{\lambda \pi d^{2}}\right)^{\!\!\! \frac{\alpha}{2}}  \right)^{N_{\rm r}}}\right)e^{-u }  \d u\nonumber \\
&\stackrel{(c)}{=}\frac{\alpha}{2\ln(2)}\int_{0}^{\infty}\frac{e^{-\frac{\left( \sinc\left(\frac{2}{\alpha}\right)u    \right)^{\! \frac{\alpha}{2}}  }{(\lambda \pi)^{\frac{\alpha}{2}}{\rm SNR}}-u}}{u}  \frac{\sum_{n=1}^{N_{\rm r}} \binom{N_{\rm{r}}}{n} \left(\frac{\sinc\left(\frac{2}{\alpha}\right)u }{\lambda \pi d^{2}}\right)^{\!\!\! n\frac{\alpha}{2}}   }{\left(1+\left(\frac{\sinc\left(\frac{2}{\alpha}\right)u }{\lambda \pi d^{2}}\right)^{\!\!\! \frac{\alpha}{2}}  \right)^{N_{\rm r}}}  \d u,
\label{eq:condi_ergodic_rate_2}
\end{align}
where (a) follows from the variable change:
\begin{equation}
u=\frac{\lambda \pi d^{2} }{\sinc\left(\frac{2}{\alpha}\right)} z^{\frac{2}{\alpha}},
\end{equation} (b) is due to the fact that $H_{k,k}$ is distributed like a Chi-squared with $2N_{\rm r}$ degrees of freedom, and (c) follows from the binomial expansion.
Using the distribution of $d_{k,k}$, which is uniformly distributed in the area of annulus with inner radius 1 and outer radius $R_{\rm d}$, we obtain\begin{align}
&\mathbb{E}\left[\log_2\left(1+\frac{H_{k,k}d_{k,k}^{-\alpha}}{I_k+\frac{1}{{\rm SNR}}}\right)\right] \nonumber \\
&=\frac{\alpha}{2\ln(2)} \int_{1}^{R_{\rm d}} \int_{0}^{\infty}\frac{e^{-\frac{\left( \sinc\left(\frac{2}{\alpha}\right)u    \right)^{\! \frac{\alpha}{2}}  }{(\lambda \pi)^{\frac{\alpha}{2}}{\rm SNR}}-u}}{u}  \frac{\sum_{n=1}^{N_{\rm r}} \binom{N_{\rm{r}}}{n} \left(\frac{\sinc\left(\frac{2}{\alpha}\right)u }{\lambda \pi d^{2}}\right)^{\!\!\! n\frac{\alpha}{2}}   }{\left(1+\left(\frac{\sinc\left(\frac{2}{\alpha}\right)u }{\lambda \pi d^{2}}\right)^{\!\!\! \frac{\alpha}{2}}  \right)^{N_{\rm r}}}  \d u \frac{2r}{R_{\rm d}^2-1}\d r.\label{eq:condi_ergodic_rate_final}
\end{align}
This completes the proof.

%%%%%%%%%%%%%%%%%%%%%%%%%%%
%%%%%%%%% Lemma  2
%%%%%%%%%%%%%%%%%%%%%%%%%%%

\section{Proof of Lemma \ref{lem}}\label{proof:lem2}
\begin{proof} The proof relies on Jensen's inequality. 
We first focus on the lower bound. Using the facts that $X$ and $Y$ are independent and $\log_2\left(1+\frac{a}{Y}\right)$ for all $a>0$ is a convex function with respect to $Y$, we obtain a lower bound as
\begin{align}
&\mathbb{E}_{X,Y}\left[\log_2\left(1+\frac{X}{Y}\right)\right] \geq \mathbb{E}_X\left[\log_2\left(1+\frac{X}{\mathbb{E}[Y]}\right)\right].
\end{align}
Since $\log_2\left(1+be^{X}\right)$ is a convex function with respect to $X$ for $b>0$, we apply Jensen's inequality again, which yields
\begin{align}
&\mathbb{E}_{X,Y}\left[\log_2\left(1+\frac{X}{Y}\right)\right] \geq \log_2\left(1+\frac{\exp\left(\mathbb{E}[\ln(X)]\right)}{\mathbb{E}[Y]}\right). \label{eq:lower_bound_lemma}
\end{align}
This completes the proof of the lower bound.

Next, we prove the upper bound. Since $\log_2(1+aX)$ is a concave function with respect to $X>0$ for all $a>0$, we obtain the upper bound
\begin{align}
&\mathbb{E}_{X,Y}\left[\log_2\left(1+\frac{X}{Y}\right)\right] \leq  \log_2\left(1+\mathbb{E}[X]\mathbb{E}\left[\frac{1}{Y}\right]\right),\label{eq:up_1}
\end{align}
which completes the proof. 

%Further, using the fact that $\log_2\left(1+\frac{a}{-\ln\left(Y\right)}\right)$ is a concave function with respect to $Y>0$ for $a>0$, we can rewrite the upper bound in (\ref{eq:up_1}) in the following  integral form:
%\begin{align}
% \mathbb{E}_Y\left[\log_2\left(1+\frac{\mathbb{E}[X]}{Y}\right)\right] &= \frac{1}{\ln(2)} \mathbb{E}_Y\left[\int_{0}^{\infty}\frac{e^{-z}}{z}(1-e^{-z\frac{\mathbb{E}[X]}{Y}})\d z\right] \nonumber \\
%&= \frac{1}{\ln(2)} \mathbb{E}_Y\left[\int_{0}^{\infty}\frac{e^{-s\frac{Y}{\mathbb{E}[X]}}}{s}(1-e^{-s})\d s\right],
%\end{align}
% 
%\begin{align}
%&= \frac{1}{\ln(2)} \mathbb{E}_Y\left[\ln\left(1+\frac{Y}{\mathbb{E}[X]}\right)\right] \\ \nonumber 
%&\leq \mathbb{E}_X\left[\log_2\left(1+\frac{X}{\mathbb{E}[Y]}\right)\right],
%\end{align}

 \end{proof}

%%%%%%%%%%%%%%%%%%%%%%%%%%%%%
%%%%%%% Theorem 3
%%%%%%%%%%%%%%%%%%%%%%%%%%%%%

\section{Proof of Theorem \ref{Th3}}\label{proof:Th3}

We prove Theorem \ref{Th3} by leveraging Lemma \ref{lem1}. From Lemma \ref{lem1}, conditioned $d_{k,k}=d$, we rewrite the sum spectral efficiency in terms of the Laplace transforms of $\tilde{H}_{k,k}$ and $\tilde{I}_k$ as
\begin{align}
&\mathbb{E}\left[\log_2\left(1\!+\!\frac{{\tilde H}_{k,k}d^{-\alpha}}{{\tilde I}_k+\frac{1}{{\rm SNR}}}\right)\!\mid \{d_{k,k}\!=\!d \}\right]\nonumber \\
&\!\!=\!\!\frac{1}{\ln(2)}\int_{0}^{\infty}\!\frac{e^{-\frac{z}{{\rm SNR}}}}{z}\left(1-\mathbb{E}\left[e^{-z  \tilde{H}_{k,k}d^{-\alpha} } \mid \{d_{k,k}\!=\!d \}\right]\right)  \mathbb{E}\left[e^{-z \sum_{j={L}+1}^{\infty}\tilde{H}_{k,k_j}d_{k,k_j}^{-\alpha} } \right] \d z. \label{eq:ergodic_expansion2}
\end{align}
Since the Laplace transform of $\tilde{H}_{k,k}$ is 
\begin{align}
\mathbb{E}\left[e^{-z  \tilde{H}_{k,k}d^{-\alpha} } \mid \{d_{k,k}\!=\!d \}\right]=\frac{1}{(1+d^{\alpha}z)^{N_{\rm r}}},
\end{align}
 we only need to compute the Laplace transform of $\tilde{I}_k$ for the given $L$, which yields,
\begin{align}
\mathcal{L}_{\tilde{I}_k}(L;z)=\mathbb{E}\!\left[e^{-z \sum_{j
={L}+1}^{\infty}\!H_{k,k_j}d_{k,k_j}^{-\alpha} }\right]\!.
\label{eq:LF_D}
\end{align}
Conditioning on the fact that the $L$th nearest interferer's distance is equal to $r$, i.e., $d_{k,k_{L}}\!\!=\!r$, the Laplace transform is computed as 
\begin{align}
\mathcal{\tilde L}_{\tilde{I}_k}(L;z)\!\!&=\mathbb{E}\left[e^{-z  \sum_{j=L+1}^{\infty}\tilde{H}_{k,k_j}d_{k,k_j}^{-\alpha}} \mid \{  d_{k,k_{\ell}}=\!r\}\right] \nonumber \\
&=  \mathbb{E}\left[ \prod_{d_{k,k_j}\in \Phi/\mathcal{B}(0,r)} \frac{1}{1+z  d_{k,k_j}^{-\alpha}}\mid \{ d_{k,k_{\ell}}=r\}\!\right]\nonumber \\
&= \exp\left({ -  \pi\lambda  \int_{r^2}^{\infty} \frac{1}{1+z^{-1}u^{\frac{\alpha}{2}}} \d u    }\right) .\label{eq:conditional_LF}
\end{align}
By unconditioning (\ref{eq:conditional_LF}) with respect to $r$ using the distribution in \cite{Haenggi}, we obtain the Laplace transform of the aggregate interference power as
\begin{align}
&\mathcal{L}_{\tilde{I}_k}(L;z)=\mathbb{E}_{d_{k,k_{L}}} \left[\mathcal{ \tilde L}_{\tilde{I}_k}(L;z)\right] \nonumber\\  
&=\!\int_{0}^{\infty}\!\!\! e^{ - \pi\lambda  \int_{r^2}^{\infty} \frac{1}{1+z^{-1}u^{\frac{\alpha}{2} } } \d u    }   \frac{2(\lambda\pi r^2)^{L}}{r \Gamma(L)} e^{-\lambda \pi r^2} \d r. \label{eq:LF_I}
\end{align}
Invoking (\ref{eq:LF_D}) and (\ref{eq:LF_I}) into (\ref{eq:ergodic_expansion2}), the conditional spectral efficiency of the $k$th link can be rewritten as 
\begin{align}
&\mathbb{E}\left[\log_2\left(1\!+\!\frac{\tilde{H}_{k,k}d^{-\alpha}}{\tilde{I}_k+\frac{1}{{\rm SNR}}}\right)\mid \{d_{k,k}\!=\!d\}\right]\nonumber \\
&\!\!=\!\!\frac{1}{\ln(2)}\int_{0}^{\infty}\!\frac{e^{-\frac{z}{{\rm SNR}}}}{z}\left\{1-\mathcal{L}_{\tilde{H}_{k,k}}(d;z)\right\}\mathcal{L}_{\tilde{I}_k}(L;z)\d z. \label{eq:ergodic_SIC_1}
\end{align}
Using the fact that $d_{k,k}$ is uniformly distributed in a disk with radius $R_{\rm d}$, i.e., $f_{d_{k,k}}(x)=\frac{2x}{R_{\rm d}^2}$ for $1\leq x\leq R_{\rm d}$, we finally obtain the result in Theorem \ref{Th3}, which completes the proof.

%%%%%%%%%%%%%%%%%%%%%%%%%%%%
%%%%%%%%   Lemma 3   %%%%%%%%%%%%%
%%%%%%%%%%%%%%%%%%%%%%%%%%%%

\section{Proof Lemma \ref{lem:correlation}} \label{proof:lem_corr}
\proof When the $k$th receiver employs the MRC strategy ${\bf w}_k^*=\frac{{\bf h}_{k,k}}{\|{\bf h}_{k,k}^*\|_2}$, the effective channel gain can be written as
\begin{align}
H_{k,k} &= \|{\bf w}_k^*{\bf h}_{k,k}\|_2^2 \nonumber \\
&=\frac{|{\bf \tilde h}_{k,k}^*{\bf C}^{1/2}{\bf C}^{1/2}{\bf \tilde h}_{k,k}|^2}{\|{\bf C}^{1/2}{\bf \tilde h}_{k,k}\|_2^2} \nonumber \\
&={\bf \tilde h}_{k,k}^*{\bf C}{\bf \tilde h}_{k,k}. \label{eq:H1}
\end{align}
Using the egienvalue decomposition of ${\bf C}$, which gives ${\bf C}={\bf U}{\bf \Lambda}{\bf U}^*$, (\ref{eq:H1}) can be rewritten as 
\begin{align}
H_{k,k}&={\bf \tilde h}_{k,k}^*{\bf U}{\bf \Lambda}{\bf U}{\bf \tilde h}_{k,k} \nonumber \\
&={\bf \bar{h}}_{k,k}^*{\bf \Lambda}{\bf \bar{h}}_{k,k} \nonumber \\
&=\sum_{n=1}^{r}\mu_n |{ \bf \bar{h}}_{k,k}(n)|^2,
\end{align}
where the second equality follows from the definition of ${\bf \bar{h}}_{k,k}={\bf U}{\bf \tilde h}_{k,k}$. Since a unitary transform does not change the distribution of elements, the $n$th element of ${\bf \bar{h}}_{k,k}$, ${ \bf \bar{h}}_{k,k}(n)$, is also $\mathcal{CN}(0,1)$. The last equality follows from the fact that ${\bf \Lambda}$ is a diagonal matrix with the entries $\left\{ \mu_1,\ldots,  \mu_{r} \right\}$. As a result,  $H_{k,k}$ is distributed as the sum of exponential random variables with means $\left\{ \mu_1,\ldots, \mu_{r} \right\}$.

Next, we characterize a simple upper bound on the distribution of $H_{k,\ell}$. With the MRC decoding strategy, the fading power for the interfering link $H_{k,\ell}$ is
\begin{align}
H_{k,\ell} &= \|{\bf w}_{k}^*{\bf h}_{k,\ell}\|_2^2 \nonumber \\
&=\frac{|{\bf \tilde h}_{k,k}^*{\bf C}^{\frac{1}{2}}{\bf C}^{\frac{1}{2}}{\bf \tilde h}_{k,\ell}|^2}{\|{\bf C}^{\frac{1}{2}}{\bf \tilde h}_{k,k}\|_2^2} \nonumber \\
&=\frac{|{\bf \bar h}_{k,k}^*{\bf \Lambda} {\bf \bar h}_{k,\ell}|^2}{{\bf \bar h}_{k,k}^*   {\bf \Lambda}{\bf \bar h}_{k,k}}, \label{eq:H2}%\nonumber \\
%&=\frac{|{\bf \bar h}_{k,k}^*{\bf \Lambda} {\bf \bar h}_{k,\ell}|^2}{{\bf \bar h}_{k,k}^*   {\bf \Lambda}{\bf \bar h}_{k,k}} 
\end{align}
where the last equality follows from the change of basis ${\bf \bar h}_{k,k}={\bf U}{\bf \tilde h}_{k,k}$ and the distribution invariance of the unitary transformation. By selecting ${\bf \bar h}_{k,k}$ as the unit norm vectors  $[1~0~\cdots~0]^T$, this fading power is upper bounded as 
\begin{align}
 H_{k,\ell} \leq  \mu_{1}|{\bf \bar h}_{k,\ell}(1)|^2.
\end{align}
Since $|{\bf \tilde h}_{k,\ell}(1)|^2$ is distributed as an exponential random variable with mean one, the complementary cumulative distribution function (CCDF) of the fading power $H_{k,\ell}$ is upper bounded as
\begin{align}
\mathbb{P}\left[H_{k,\ell}>x\right] \leq \exp\left(-\frac{x}{\mu_{1}}\right).
\end{align}
Consequently, under antenna correlation scales, the mean of $H_{k,\ell}$ is upper bounded by the maximum eigenvalue of the correlation matrix. This completes the proof.
\endproof
 
{}

\end{document}